\documentclass[lettersize,journal]{IEEEtran}
\usepackage{amsmath,amsfonts}
\usepackage{algorithm}
\usepackage{algpseudocode}

\usepackage{array}
\usepackage{textcomp}
\usepackage{stfloats}
\usepackage{url}
\usepackage{verbatim}

\usepackage{graphicx}
\usepackage[caption=false]{subfig}

\usepackage{threeparttable}
\usepackage{caption}
\usepackage{stmaryrd}

\usepackage{makecell}
\usepackage{multirow}
\usepackage{booktabs}

\usepackage{siunitx}
\usepackage{adjustbox}

\usepackage{amsthm}
\newtheorem{theorem}{Theorem}

\usepackage{amssymb}
\usepackage{amsmath}
\usepackage{enumitem}

\usepackage{cite}
\usepackage{xcolor}

\begin{document}

\title{
Enhancing Security and Privacy in Federated Learning using Low-Dimensional Update Representation and Proximity-Based Defense
}

\author{

\IEEEauthorblockN{Wenjie Li,
Kai Fan,~\IEEEmembership{Member, IEEE},
Jingyuan Zhang, 
Hui Li,~\IEEEmembership{Member,~IEEE},\\
 Wei Yang Bryan Lim,
 and Qiang Yang,~\IEEEmembership{Fellow,~IEEE}}

\thanks{This work was supported in part by the National Natural Science Foundation of China (No.62372356 and No.92067103), the Key Research and Development Program of Shaanxi (2021ZDLGY06-02), the Natural Science Foundation of Shaanxi Province (No.2019ZDLGY1202), the Shaanxi Innovation Team Project (No.2018TD-007), the Xi’an Science and Technology Innovation Plan (No.20189168CX9JC10), the 111 Center (B16037), the NTU Startup Grant, RIE2025 Industry Alignment Fund – Industry Collaboration Projects (IAF-ICP) (Award I2301E0026), administered by A*STAR, as well as supported by Alibaba Group and NTU Singapore through Alibaba-NTU Global e-Sustainability CorpLab (ANGEL). (Corresponding author: Kai Fan)}
\thanks{Wenjie Li is with the State Key Laboratory of Integrated Service Networks, Xidian University, Xi'an 710071, China and also with the College of Computing and Data Science, Nanyang Technological University, Singapore 639798 (e-mail: tom643190696@gmail.com).} \thanks{Kai Fan and Hui Li are with the State Key Laboratory of Integrated Service Networks, Xidian University, Xi'an 710071, China (e-mail: kfan@mail.xidian.edu.cn; lihui@mail.xidian.edu.cn).}  \thanks{Jingyuan Zhang and Wei Yang Bryan Lim are with the College of Computing and Data Science, Nanyang Technological University, Singapore 639798 (e-mail: jzhang149@e.ntu.edu.sg; bryan.limwy@ntu.edu.sg).} \thanks{Qiang Yang is with the Department of Computer Science and Engineering,
Hong Kong University of Science and Technology, Hong Kong 999077, China (e-mail: qyang@cse.ust.hk).}

}


\markboth{Journal of \LaTeX\ Class Files,~Vol.~14, No.~8, August~2021}%
{Shell \MakeLowercase{\textit{et al.}}: A Sample Article Using IEEEtran.cls for IEEE Journals}


\maketitle

\begin{abstract}

Federated Learning (FL) is a promising privacy-preserving machine learning paradigm that allows data owners to collaboratively train models while keeping their data localized. Despite its potential, FL faces challenges related to the trustworthiness of both clients and servers, particularly against curious or malicious adversaries. In this paper, we introduce a novel framework named \underline{F}ederated \underline{L}earning with Low-Dimensional \underline{U}pdate \underline{R}epresentation and \underline{P}roximity-Based defense (FLURP), designed to address privacy preservation and resistance to Byzantine attacks in distributed learning environments. FLURP employs $\mathsf{LinfSample}$ method, enabling clients to compute the $l_{\infty}$ norm across sliding windows of updates, resulting in a Low-Dimensional Update Representation (LUR). Calculating the shared distance matrix among LURs, rather than updates, significantly reduces the overhead of Secure Multi-Party Computation (SMPC) by three orders of magnitude while effectively distinguishing between benign and poisoned updates. Additionally, FLURP integrates a privacy-preserving proximity-based defense mechanism utilizing optimized SMPC protocols to minimize communication rounds. Our experiments demonstrate FLURP's effectiveness in countering Byzantine adversaries with low communication and runtime overhead. FLURP offers a scalable framework for secure and reliable FL in distributed environments, facilitating its application in scenarios requiring robust data management and security. 

\end{abstract}


\section{Introduction}

\IEEEPARstart{F}{ederated} Learning (FL)~\cite{fedavg-pmlr-v54-mcmahan17a} is a distributed machine learning paradigm that enables data owners to collaboratively train models while keeping their data local, sharing only model updates with a central server. The server integrates these updates based on an Aggregation Rule (AR), thus facilitating collaborative model training. Notably, models trained via FL can achieve accuracies comparable to those trained with centralized methods~\cite{asgoodas-CAI2024102411}.

However, the effectiveness of FL critically depends on the trustworthiness of both the data owners (i.e., clients) and the server. The presence of curious or malicious adversaries in FL introduces a complex interplay between clients and servers~\cite{shieldfl-ma2022shieldfl, pbfl-li2023pbfl, fltrus-DBLP:conf/ndss/CaoF0G21, flod-DBLP:conf/esorics/DongCLWZ21}. In classic AR such as Federated Averaging (FedAvg)~\cite{fedavg-pmlr-v54-mcmahan17a}, the server performs the weighted averaging-based aggregation on updates from clients indiscriminately. This process is vulnerable to Byzantine adversaries that may tamper with local data or modify the training process to generate poisoned local model updates~\cite{minmax-shejwalkar2021manipulating}. Studies~\cite{backdoor-DBLP:journals/corr/abs-1911-07963, backdoor-pmlr-v108-bagdasaryan20a, multikrum-NIPS2017_f4b9ec30} have shown that even a single malicious client can cause the FedAvg-derived global model to deviate from the intended training direction. Byzantine adversaries may also execute attacks sporadically~\cite{multikrum-NIPS2017_f4b9ec30, lf-jere2020taxonomy}, including backdoor attacks that may require only a single round to embed a persistent backdoor in the global model~\cite{backdoor-DBLP:journals/corr/abs-1911-07963, backdoor-pmlr-v108-bagdasaryan20a}. Since the server cannot access the client's data, Byzantine-resistant aggregation rules~\cite{multikrum-NIPS2017_f4b9ec30, trimmedmean-pmlr-v80-yin18a, fltrus-DBLP:conf/ndss/CaoF0G21, fltracer-zhang2023fltracer} need to be implemented.

Moreover, the server may compromise clients' data privacy, e.g., through data reconstruction attacks~\cite{DLG-NEURIPS2019_60a6c400, iDLG-DBLP:journals/corr/abs-2001-02610, InvertGrad-NEURIPS2020_c4ede56b} or source attacks~\cite{Sourceinferenceattack-hu2021source} simply by analyzing the plaintext updates from clients. For example, DLG~\cite{DLG-NEURIPS2019_60a6c400} initially generates a random input and then calculates the $l_2$ distance between the fake gradient and the real plaintext gradient as a loss function. InvertGrad~\cite{InvertGrad-NEURIPS2020_c4ede56b} employs cosine distance and total variation as loss function. As such, clients need to encrypt or perturb the uploaded updates. However, considering the false sense of security~\cite{xujiaanquanganyuyigongji-DBLP:conf/uss/YueJWBD23} brought by perturbation and quantization, mechanisms such as Secure Multi-Party Computation (SMPC)~\cite{assABY-demmler2015aby, asscryptflow2-rathee2020cryptflow2} are necessary to safeguard the privacy of updates.

To address the aforementioned issues, we need a privacy-preserving but efficient defense to maintain the privacy of clients and robustness of FL system. We find that using sliding windows to calculate the $l_{\infty}$ vector of model updates can still capture subtle differences between benign updates and those that are poisoned, even when the latter inject deviations based on the element-wise mean and variance of benign updates or are trained on tampered samples. Consequently, we propose using this $l_{\infty}$ vector, termed Low-Dimensional Update Representation (LUR), to compute the pairwise squared Euclidean distance (SED) instead of directly using model updates, significantly reducing the computational overhead associated with SMPC on servers. Furthermore, we introduce a novel proximity-based defense to eliminate poisoned updates, implemented through a series of SMPC protocols that systematically convert the sharing of the SED matrix into client qualifications. The key contributions of this paper are as follows:


\begin{itemize}
    \item We introduce FLURP, a novel framework that integrates the LUR calculation method, $\mathsf{LinfSample}$, and a privacy-preserving proximity-based defense to mitigate poisoned model updates. FLURP allows clients to perform minimal additional computations and upload LURs, significantly reducing the server's computational overhead for SED by three orders of magnitude, while simultaneously enhancing the security and privacy of FL.
    \item We develop optimized privacy-preserving SMPC protocols tailored for FLURP, including packed comparison techniques for \textit{Arithmetic Sharing}, homomorphic encryption-based multi-row shuffle protocols, and fast selection protocols using packed comparison. Using parallelism, these protocols minimize communication rounds, substantially improving system efficiency and scalability.
    \item Experimental results empirically demonstrate FLURP's resilience against eight state-of-the-art adversarial strategies, across varying malicious client ratios and both image and text datasets. The proposed SMPC protocols offer provable security guarantees with linearly scaled communication and computational costs.
\end{itemize}

\section{Related Work}

\subsection{Byzantine Resilient Aggregation Rule (BRAR)}

Extensive research has focused on designing aggregation rules to defend against Byzantine adversaries. These methods typically assume fully trustworthy servers with unrestricted access to plaintext local model updates.

In scenarios where prior knowledge is limited, such as the number of malicious clients (denoted as $\beta$), the server computes statistical information of updates, including geometric distance, median, and norm. The Multikrum~\cite{multikrum-NIPS2017_f4b9ec30} selects updates with the smallest sum of distances to the other $m-\beta-2$ updates, where $m$ is the number of clients. The Trimmedmean~\cite{trimmedmean-pmlr-v80-yin18a} removes the $\beta$ extreme values from each dimension of the updates and computes the average of the remaining values. However, these methods, relying on dimensional-level calculations and distance measures, have been shown to be vulnerable to complex attacks~\cite{minmax-shejwalkar2021manipulating, li2024blades}.

Other aggregation strategies utilize clustering and clipping to detect poisoned updates from multiple perspectives. The ClippedClustering~\cite{cc-10018261} adaptively clips updates based on the median of historical norms and employs agglomerative hierarchical clustering to partition updates into two clusters, averaging those in the larger cluster. The DnC~\cite{minmax-shejwalkar2021manipulating} uses singular value decomposition to detect outliers and applies random sampling to reduce dimensionality. The MESAS~\cite{MESAS-CCS23} selects six metrics of local models to filter poisoned updates effectively.

Trust-based methods rely on the server's collection of a clean root dataset to bootstrap trust, or assume clients will honestly propagate the local model based on this dataset. The FLTrust~\cite{fltrus-DBLP:conf/ndss/CaoF0G21} and the FLOD~\cite{flod-DBLP:conf/esorics/DongCLWZ21} use this approach to defend against a significant fraction of malicious clients. These methods discard updates with directions greater than 90 degrees from the clean model, then aggregate the remaining updates based on cosine similarity. However, FLOD~\cite{flod-DBLP:conf/esorics/DongCLWZ21} increases the number of rounds required for convergence and reduces the global model's accuracy by using signed updates for aggregation. FLARE~\cite{flare-wang2022flare} requires forward propagation on a designated dataset, uploading penultimate-layer representations, and calculates the maximum mean discrepancy between these representations as a distance measure for model aggregation.

\subsection{Privacy-Preserving Byzantine Resilient Aggregation Rule (ppBRAR)}

Recent studies have sought to simultaneously mitigate privacy leakage from aggregators targeting plaintext model updates and defend against Byzantine adversaries. These approaches adapt BRAR to ppBRAR~\cite{pbfl-li2023pbfl, flod-DBLP:conf/esorics/DongCLWZ21, shieldfl-ma2022shieldfl} by utilizing SMPC. This transformation replaces the centralized server with multiple distributed servers that receive encrypted model updates, which are then collaboratively processed to execute the ppBRAR protocol.

The PPBR~\cite{ppbrf-dongcai2023privacy} designs a 3-party computation (3PC)-based maliciously secure top-$k$ protocol that calculates the sum of the largest $k$ similarities for each update, selecting those with the highest scores for weighted aggregation. The ShieldFL~\cite{shieldfl-ma2022shieldfl} leverages double-trapdoor homomorphic encryption to compute cosine similarities between local updates and the global update of the previous round, assigning lower weights to updates with the lowest similarity. However, it assumes benign behavior in the first round, which is unrealistic in the case of a static adversary. The LFR-PPFL~\cite{iotnoniidlf-shen2023privacy} tracks distance changes between updates to identify malicious clients, but it assumes a static adversary that behaves maliciously in every round. The RoFL~\cite{rofl-DBLP:conf/sp/LycklamaBVKH23} and ELSA~\cite{ELSA-DBLP:conf/sp/RatheeSWP23} limit the norms of updates to reduce the impact of poisoned updates, with the former using zero-knowledge proofs to ensure norm correctness, and the latter using multi-bit Boolean sharing.

While BRARs fail to address privacy leakage from plaintext updates, ppBRARs, which lack prior knowledge (e.g., a clean dataset), generally require computing pairwise distances or similarities between model updates. Privately calculating these metrics for $m$ updates under SMPC involves $\binom{m}{2}$ full model-sized multiplications~\cite{flod-DBLP:conf/esorics/DongCLWZ21, ppbrf-dongcai2023privacy, pbfl-li2023pbfl}. For models with millions of parameters, these operations dominate the overall time complexity~\cite{pbfl-li2023pbfl}. Additionally, clustering-based ppBRAR methods~\cite{pbfl-li2023pbfl, RFBDS-chen2023privacy} require numerous interaction rounds. Therefore, minimizing the overhead of pairwise distance computation and poisoned update filtering becomes a critical challenge.

\section{Preliminaries}

\subsection{Federated Learning}
 A classical FL algorithm typically includes an initial step 1, followed by $T$ repetitions of steps 2-4.

\begin{itemize}
  \item \textbf{Step 1:} The server initializes a global model randomly.
  \item \textbf{Step 2:} At the global round $t$, the server selects a set of clients $\mathcal{C}_t$ and broadcasts the global model $\boldsymbol{\omega}^t$.
  \item \textbf{Step 3:} Each client $i$ locally performs the replacement $\boldsymbol{\omega}_i^t\leftarrow \boldsymbol{\omega}^t$ and estimates the gradients of the loss function over their private dataset for $E$ epochs. Then $i$ computes and uploads the local model update $\boldsymbol{g}_i^t \leftarrow \boldsymbol{\omega}_i^t - \boldsymbol{\omega}^t$.
  \item \textbf{Step 4:} The server collects the local model updates and utilizes an AR to compute the global update, which is then used to update the global model.
\end{itemize}

Note that when the aggregating party implements ppBRAR, the centralized server is replaced by multiple distributed servers.

\subsection{Byzantine Attacks in FL}\label{sec:byzantine}

The Byzantine adversary encompasses a spectrum of malicious behaviors, whereby malicious clients manipulate local data or the training process to generate poisoned updates, in an attempt to disrupt the training process of the global model. Specifically, these attacks can be classified into \textit{Untargeted} and \textit{Targeted} attacks based on their objectives. We denote the set of clients corrupted by the adversary as $\mathcal{C'}$.

\subsubsection{Untargeted attacks}

In untargeted attacks, adversaries construct poisoned updates that, when aggregated, cause the model to lose usability or delay convergence. The feasible attack types depend on the adversary's knowledge and capabilities. Adversary with limited knowledge of malicious clients' data and training can launch:
\begin{itemize}
        \item \textbf{LabelFlipping}~\cite{lf-jere2020taxonomy}. The adversary flips labels of all local samples. In a dataset with $c$ classes, each label $y$ is transformed to $c - 1 - y$. 
        \item \textbf{SignFlipping}~\cite{sf-an-li2019rsa}. The adversary inverts the signs of all gradient entries before updating the local model, essentially reversing gradient descent.
        \item \textbf{Noise}-$(\mu, \sigma^2)$~\cite{sf-an-li2019rsa}. The adversary uploads a vector $\boldsymbol{x}$ with elements sampled from $\mathcal{N}(\mu, \sigma^2)$. We use the standard normal distribution, denoted as Noise-(0,1).
\end{itemize}

The adversary, with knowledge of all benign client updates, constructs sophisticated attacks by calculating the element-wise mean ($\boldsymbol{\mu}$) and standard deviation ($\boldsymbol{\sigma}$) of these updates. They can launch:
\begin{itemize}   
        \item \textbf{ALIE}~\cite{alie-baruch2019little}. The adversary calculates $\alpha$ using the Inverse Cumulative Distribution Function of a standard normal distribution. For each $i \in \mathcal{C'}$, $\boldsymbol{g}_{i} = \boldsymbol{\mu} + \boldsymbol{\sigma} \cdot \alpha$, introducing an offset. 
        \item \textbf{MinMax}~\cite{minmax-shejwalkar2021manipulating}. The adversary computes a scaling factor $\alpha$ to ensure the maximum distance from any poisoned update $\boldsymbol{g}_{i}$ to any benign update does not exceed the maximum distance between benign updates. Each poisoned update is $\boldsymbol{g}_{i} = \boldsymbol{\mu} - \boldsymbol{\sigma} \cdot \alpha$.
        \item \textbf{IPM}-$\alpha$~\cite{ipm-pmlr-v115-xie20a}. Malicious clients upload an average update in the opposite direction, scaled by $\alpha$. For each $i \in \mathcal{C'}$, $\boldsymbol{g}_i = - \alpha \cdot \boldsymbol{\mu}$. We implement two variants: IPM-0.1 ($\alpha = 0.1$) and IPM-100 ($\alpha = 100$).
\end{itemize}

\subsubsection{Targeted attack~/~\textbf{Backdoor} attack\cite{backdoor-DBLP:journals/corr/abs-1911-07963, backdoor-pmlr-v108-bagdasaryan20a, yimingli9802938}}
In a backdoor attack, the adversary embeds a trigger (e.g., a visible red square or subtle perturbation) into features and changes their labels to a specific target. The resulting poisoned model performs correctly on clean samples but misclassifies triggered samples as the target label. In our implementation process, the manipulated clients insert a red square pattern in the upper-left corner of half the samples (for image samples) and prepend the low-frequency word "tt" (for text samples). Subsequently, these clients set the target label to 0.

\subsection{Additive Secret Sharing} \label{subsec:secretshare}

Additive secret sharing~\cite{assABY-demmler2015aby} splits a secret into multiple random shares, each held by different parties, with no individual share revealing any information. The secret is reconstructed by summing the shares. We use two cloud servers, $S_0$ and $S_1$, to perform linear computations on the shares. For an $l$-bit \textit{Arithmetic sharing} $\langle x \rangle$ of $x$, the value of $x$ is given by $x = \langle x \rangle_0 + \langle x \rangle_1 \mod 2^l$, with $\langle x \rangle_0, \langle x \rangle_1 \in \mathbb{Z}_{2^l}$. When $l=1$, the sharing of $x \in \{0,1\}$ is termed \textit{Boolean sharing} $\langle x \rangle^B$. The foundational operations for sharing are defined as follows:
 
\textbf{Splitting.} $\mathsf{Split}(x)$: The holder of a secret $x \in \mathbb{Z}_{2^l}$ chooses $r \in_{R} \mathbb{Z}_{2^l}$ and computes shares $\langle x \rangle_0 = r$ and $\langle x \rangle_1 = x - r \mod 2^l$. The shares are then distributed to $S_0$ and $S_1$. To enhance efficiency, a pre-negotiated seed is used for generating $\langle x \rangle_0$, effectively halving communication overhead.

\textbf{Revealing.} $\mathsf{Reavel}(\langle x \rangle)$: To recover the secret $x$, party $S_b$ sends its share $\langle x \rangle_b$ to the other party $S_{1-b}$, which then computes $x = \langle x \rangle_0 + \langle x \rangle_1 \mod 2^l$.

\textbf{Addition.} $\mathsf{ADD}(\langle x \rangle, \langle y \rangle)$: To compute the sum $\langle z \rangle = \langle x + y \rangle$, $S_b$ locally computes $\langle z \rangle_b = \langle x \rangle_b + \langle y \rangle_b \mod 2^l$.

\textbf{Multiplication.} $\mathsf{MUL}(\langle x \rangle, \langle y \rangle)$: To compute the product $\langle z \rangle = \langle x \cdot y \rangle$, parties utilize the pre-generated \textit{Beaver triple} ($\langle a \rangle$, $\langle b \rangle$, and $\langle c \rangle$, where $c = a \cdot b$) to calculate $e = x - a$ and $f = y - b$. Then, $S_0$ computes $\langle z \rangle_0 = ef + e \langle b \rangle_0 + f \langle a \rangle_0 + \langle c \rangle_0 \mod 2^l$ locally, while $S_1$ computes $\langle z \rangle_1 = e \langle b \rangle_1 + f \langle a \rangle_1 + \langle c \rangle_1 \mod 2^l$ locally.

\textbf{Conversion.} $\mathsf{B2A}(\langle x \rangle^B) \rightarrow \langle x \rangle$: For details on converting \textit{Boolean sharing} $\langle x \rangle^B$ to \textit{Arithmetic sharing} $\langle x \rangle$, refer to \cite{asscryptflow2-rathee2020cryptflow2}.

Specifically, $\mathsf{XOR}$ is used for addition and $\mathsf{AND}$ for multiplication in \textit{Boolean sharing}.

\subsection{Additive Homomorphic Encryption}

AHE allows operations on ciphertexts which correspond to addition or multiplication in the plaintext space, without the need to decrypt. With plaintext space $M$ and ciphertext space $C$, AHE comprises three functions and two operations:

$\mathsf{KeyGen}(1^{\kappa}) \rightarrow (sk, pk)$: Generates a key pair based on the security parameter $\kappa$. 


$\mathsf{Enc}(m, pk) \rightarrow c$: Encrypts plaintext $m \in M$ to ciphertext $c \in C$ using public key $pk$.

$\mathsf{Dec}(c, sk) \rightarrow m$: Decrypts ciphertext $c \in C$ to plaintext $m \in M$ using private key $sk$.

$\oplus$, $\boxplus$: For any plaintexts $m_1, m_2 \in M$, it holds that $\mathsf{Dec}(\mathsf{Enc}(m_1) \oplus \mathsf{Enc}(m_2)) = m_1 + m_2$ and $\mathsf{Dec}(\mathsf{Enc}(m_1) \boxplus m_2) = m_1 + m_2$. Hence, performing the $\oplus$ or $\boxplus$ operation on ciphertexts is equivalent to performing addition in the plaintext space. We employ the \textit{Paillier} scheme~\cite{paillier1999public}, with a 1024-bit plaintext space and a 2048-bit ciphertext space.

\section{Overview}

\subsection{System Model}

The system model (Figure~\ref{fig:systemmodel}) consists of two main components: the Client Side and the Aggregator Side, which cooperate to achieve the FL training task.

\textbf{Client Side}: This side comprises $m$ clients, each with its own local dataset. During each global iteration of FL, clients compute model updates based on their local data and the latest global model. Concurrently, Byzantine clients may launch poisoning attacks to generate poisoned updates. Clients create a compact representation of their updates using the sampling method $\mathsf{LinfSample}$, referred to as LURs. Both model updates and LURs are uploaded to the servers as sharing.

\textbf{Aggregator Side}: This side involves two servers responsible for executing privacy-preserving and Byzantine-robust aggregation. Upon receiving shares from clients, the servers perform a series of SMPC protocols to assess client qualifications for contributing to the global model without accessing plaintext local updates. The servers then reconstruct the global model update and broadcast the new global model. This process iterates until a satisfactory FL model is achieved, ensuring both client privacy protection and robustness against potential Byzantine attacks throughout training.

\begin{figure}[!htbp]
\centering
\includegraphics{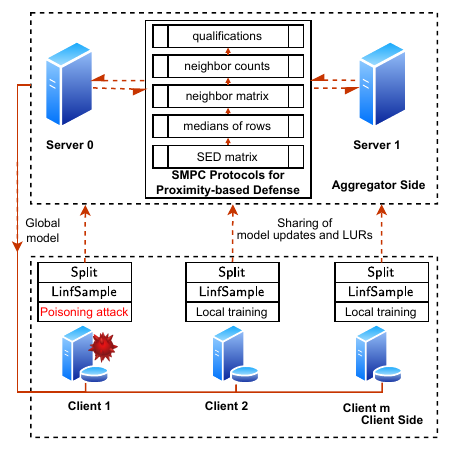}
\caption{System model}
\label{fig:systemmodel}
\end{figure}

\subsection{Threat Model}\label{sec:threatmodel}

We consider a Byzantine adversary $\mathcal{A}_f$ capable of corrupting fewer than half of the clients. This setting is consistent with previous works~\cite{li2024blades, pbfl-li2023pbfl, ppbrf-dongcai2023privacy, RFBDS-chen2023privacy}. At the beginning of FL, $\mathcal{A}_f$ controls the set of malicious clients $\mathcal{C}'$ to implement any one of eight types of attacks, including LabelFlipping, SignFlipping, Noise-(0,1), ALIE, MinMax, IPM-0.1, IPM-100, and Backdoor. These attacks are described in detail in Section~\ref{sec:byzantine}.


We consider a static, probabilistic polynomial-time adversary, $\mathcal{A}_h$, who is honest-but-curious~\cite{simulatesecurity-lindell2017simulate, asscryptflow2-rathee2020cryptflow2}. This adversarial model, often implemented using two competing cloud service providers, is widely adopted in practice due to reputational and financial incentives against collusion or protocol violations. $\mathcal{A}_h$ may control at most one server ($S_0$ or $S_1$) at the onset of FLURP. While adhering to the protocol, $\mathcal{A}_h$ attempts to infer client information from protocol transcripts, motivated by commercial interests. Under the threat of $\mathcal{A}_h$, $S_0$ and $S_1$ will not collude to recover the plaintext local model updates from the received secret shares, and therefore cannot launch sample reconstruction attacks~\cite{DLG-NEURIPS2019_60a6c400, iDLG-DBLP:journals/corr/abs-2001-02610, InvertGrad-NEURIPS2020_c4ede56b} or property inference attacks~\cite{Sourceinferenceattack-hu2021source} against specific clients. We define the security of our SMPC protocols using the ideal-world/real-world simulation paradigm~\cite{asscryptflow2-rathee2020cryptflow2}. An SMPC protocol $\Pi$ is considered secure if, within this hybrid model, $\mathcal{A}_h$ cannot distinguish between the protocol's real-world and ideal-world executions.
 

\subsection{Design Goals}

\textbf{Byzantine Robustness:} FLURP leverages the combination of SED matrices computed from LURs and proximity-based defense to defend against all eight attack types enumerated in Section~\ref{sec:threatmodel}.

\textbf{Privacy Preservation:} The aggregation process, implemented via SMPC protocols, prevents the leakage of statistical information about updates. Servers gain no additional knowledge beyond identifying which clients' updates will be aggregated.

\textbf{Efficiency:} FLURP aims to significantly reduce the overhead incurred by SMPC protocols compared to existing schemes that compute distances and perform clustering on full-sized updates \cite{RFBDS-chen2023privacy, cc-10018261}.

\textbf{Scalability:} Our tailored SMPC protocols exhibit superior communication round complexity and runtime efficiency, enabling FLURP to scale effectively as the number of clients increases.

\subsection{Defense Method with Plaintext Updates}

To better understand robust aggregation, we idealize the servers as fully trustworthy in this section, focusing solely on the adversary $\mathcal{A}_f$. We describe two key technical components used to identify benign updates.

\subsubsection{LinfSample}\label{subsec:linfsample}

Each client $i$ executes the $\mathsf{LinfSample}$ operation to compute the Low-Dimensional Update Representation (LUR) in two steps: 1. Flatten the model update: $\boldsymbol{f}_i = \mathsf{flatten}(\boldsymbol{g}_i) \in \mathbb{R}^n$. 2. Compute the $l_{\infty}$ norm of the entries selected by the sliding window to obtain the LUR $\boldsymbol{v}_i \in \mathbb{R}^{l}$:

\begin{equation}\label{eq:linfsample}
v_{i,j} = \left[ \max  \left\{ \lvert \boldsymbol{f}_{i, s \cdot j + k} \rvert \;\middle|\; k = 0, 1, ..., s-1 \right\} \right]_{j=0}^{l-1} .
\end{equation}
 
 The window size and stride are both set to $s$ with the ceiling mode. So the length of LUR is $d = \lfloor \frac{n}{s} \rfloor +  1 \{n \bmod s \neq 0\}$. The optimal value of $s$ is discussed in Section~\ref{sec:windowsize}.

\subsubsection{Proximity-based Defense}
Our threat model assumes that the majority of clients are honest. In the absence of prior knowledge and clean datasets, a common approach is to identify a cluster containing more than half of the total updates. Since updates are unlabeled, traditional k-Nearest Neighbors (k-NN) clustering cannot be applied. Therefore, we design a neighbor-counting mechanism similar to an unlabeled k-NN approach. Given $m$ participating clients, at the beginning of each global round, each client's neighbor count is initialized to $0$. The server computes the $\lfloor\frac{m}{2}\rfloor$ closest clients for each client $i$, and increments their neighbor counts by $1$. Subsequently, clients with neighbor counts of at least $\lfloor\frac{m}{2}\rfloor$ are classified as benign, and their updates are aggregated using weighted averaging to form the global update.

\section{Construction}
This section first introduces the FLURP framework, followed by the specific constructions of SMPC protocols it employs. The symbols used are listed in Table~$\ref{tab:symbol}$.

\begin{table}[!htbp]
\centering
\caption{Symbol Table}
\label{tab:symbol}
\begin{tabular}{ll|ll} 
\hline
Symbol                               & Description               & Symbol                                           & Description         \\ 
\hline
$S_0, S_1$                           & Servers                   &  $\langle \cdot \rangle^B$                       & Boolean sharing     \\
$\mathcal{C}$                        & Clients                   &  $\langle \cdot \rangle$                         & Arithmetic sharing  \\
$\boldsymbol{\omega}$                & Global model              &  $\llbracket \cdot \rrbracket$                 & AHE Ciphertext      \\
$\boldsymbol{g}_i$                   & Update of~$i$             &  $\in_{R}$                                       & Uniform sampling    \\
$\boldsymbol{v}_i$                   & LUR ($l_\infty$ vector)         &  $\kappa$                                        & Security parameter  \\
$\boldsymbol{M}$                     & SED matrix                &  $\{a_i\}_{i=1}^n$                               & Set of $a_i$        \\
$\mu_i$                              & Row median                &  $[a_i]_{i=1}^n$                                 & List of $a_i$       \\
$\boldsymbol{N}$                     & Neighbor matrix           &  $[x]_n$                                         & $n$ copies of $x$   \\
$s_i$                                & Neighbor count            &  $\{0,1\}^\kappa$                                & $\kappa$-bit strings\\
$q_i$                                & Qualification             &  $\boldsymbol{\pi}$                              & Permutation set     \\
$\mathcal{Y}$                        & Qualified clients         &  $\mathsf{reshape}(\cdot,\text{shape})$          & Reshape to shape    \\
$\tilde{\boldsymbol{M}}$             & Shuffled $\boldsymbol{M}$ &       $\mathcal{F}_{\mathsf{protocol}}$          &  Ideal function     \\
$\boldsymbol{M}^{\boldsymbol{\pi}}$  & Shuffled by ${\boldsymbol{\pi}}$  &  $\mathsf{Enc}(\cdot, pk)$               & Encryption          \\
$\bar{\boldsymbol{d}}$               & Dictionary                &  $\mathsf{Dec}(\cdot, sk)$                       & Decryption          \\
$\oplus, \boxplus$                  & AHE addition              &  $\mathsf{len}(\cdot)$                           & Sequence length     \\
$pk$                                 & Public key                &  $\mathsf{matrixShuffle}$                        & Matrix shuffle      \\
$sk$                                 & Secret key                &  $\mathsf{slice}(\cdot, \text{from}, \text{to})$ & Slice the sequence  \\
\hline
\vspace{-30pt}
\end{tabular}
\end{table}

\subsection{Framework}

We propose the FLURP framework (Algorithm~\ref{framework}) to enhance model robustness against adversarial manipulations by employing an SMPC-based ppBRAR. The framework operates with servers $S_0$ and $S_1$, $m$ clients, their local datasets, and publicly preset parameters.

On the client side, after local training, each client $i$ applies $\mathsf{LinfSample}$, as described in Section~\ref{subsec:linfsample}, to their update $\boldsymbol{g}_i$, yielding the local update representation (LUR) $\boldsymbol{v}_i$. Subsequently, both $\boldsymbol{g}_i$ and $\boldsymbol{v}_i$ are split into \textit{Arithmetic sharing} and transmitted to the servers. 

On the aggregator side, the servers execute a series of SMPC protocols to implement a privacy-preserving proximity-based defense. They compute secret shares of the following:
(1) the SED matrix of LURs,
(2) the median of each row in the SED matrix,
(3) the neighbor matrix,
(4) neighbor counts for all clients, and
(5) qualifications of all clients, which determine which clients can contribute to the global model.

Specifically, the servers compute the shared SED matrix $\langle \boldsymbol{M} \rangle$ of $[\langle \boldsymbol{v}_i \rangle]_{i=1}^{m}$. For $i, j \in [1, m]$ where $i \neq j$, they compute $\langle M_{i,j} \rangle = \langle M_{j,i} \rangle = \mathsf{MUL}(\langle \boldsymbol{v}_i \rangle - \langle \boldsymbol{v}_j \rangle, \langle \boldsymbol{v}_i \rangle - \langle \boldsymbol{v}_j \rangle)$. For cases where $i = j$, the servers set $\langle M_{i,j} \rangle_0 = \langle M_{i,j} \rangle_1 = 0$. FLURP employs the meticulously designed SMPC protocol $\mathsf{packedCompare}$ as the underlying component for comparisons across multiple pairs of shares in subsequent protocols. $\mathsf{packedCompare}$ features extremely low communication round overhead.

\begin{algorithm}[!htbp]
\caption{The FLURP framework}\label{framework}
\begin{algorithmic}[1]
  \Statex {\bf Input:} Client set $\mathcal{C} = \lbrace C_1, C_2,\ldots, C_m \rbrace$, two servers $S_0$ and $S_1$, the number of global rounds $T$, local learning rate $\eta_l$, batch size $b$, the collection of local datasets $\mathcal{D} = \lbrace \mathcal{D}_1, \mathcal{D}_2, \ldots, \mathcal{D}_m \rbrace$, the size $s$ of the sampling window used by $\mathsf{LinfSample}$
  \Statex  {\bf Output:} Global model parameter $\boldsymbol{\omega}$
  
  \State Randomly initialize the global model $\boldsymbol{\omega}_0$
  
  \For{global round $t\in [0, T-1]$}
    \State $\triangleright$ \textbf{Client side}:
        \For{each client $i \in [1, m]$ \textbf{in parallel}}
          \State Run $\text{SGD}(b, \boldsymbol{\omega}_t, \mathcal{D}_i, \eta_l)$ for $E$ epochs to get $\boldsymbol{g}_i$ 
          \State $\boldsymbol{v}_i \gets \mathsf{LinfSample}(\boldsymbol{g}_i, s)$
          \State $\langle \boldsymbol{v}_i \rangle, \langle \boldsymbol{g}_i \rangle \gets \mathsf{Split}(\boldsymbol{v}_i), \mathsf{Split}(\boldsymbol{g}_i)$  
      \EndFor
    
    \State $\triangleright$ \textbf{Aggregator side}:  
    
    \State Calculate the shared SED matrix $\langle \boldsymbol{M} \rangle$ of $[ \langle \boldsymbol{v}_i \rangle ]_{i=1}^{m}$ 

    \State $\langle \tilde{\boldsymbol{M}} \rangle \gets \mathsf{matrixSharedShuffle}(\langle \boldsymbol{M} \rangle)$
    
    \State \textbf{\#} Calculate the shared median of each row in $\langle \boldsymbol{M} \rangle$

    \State $[\langle \mu_i \rangle]_{i=1}^{m} \gets \mathsf{mulRowQuickSelect}(\langle \tilde{\boldsymbol{M}} \rangle, [\lfloor\frac{m}{2}\rfloor]_m)$
    
    \State \textbf{\#} Calculate the shared neighbor matrix $\langle \boldsymbol{N} \rangle$
       
    \State $[\langle r_i \rangle^B]_{i=1}^{m^2} \gets \mathsf{packedCompare}(\mathsf{flatten}(\langle \boldsymbol{M} \rangle),$ 
    $[ [\langle \mu_1 \rangle]_m, \ldots, [\langle \mu_m \rangle]_m ])$
    
    \State $[\langle r_i \rangle]_{i=1}^{m^2} \gets \mathsf{B2A}([\langle r_i \rangle^B]_{i=1}^{m^2})$

    \State $\langle \boldsymbol{N} \rangle \gets \mathsf{reshape}([\langle r_i \rangle]_{i=1}^{m^2}, m \times m)$

    \State \textbf{\#} Calculate the shared neighbor count $\langle s_i \rangle$ of each client $i$ 
    
    \State For $i \in [1,m]$, $\langle s_i \rangle \gets \mathsf{ADD}(\lbrace \langle {N}_{1, i} \rangle, \ldots, \langle {N}_{m, i} \rangle \rbrace)$

    \State \textbf{\#} Identify clients with neighbor count of at least $\lfloor \frac{m}{2} \rfloor$
    
    \State $[\langle q_i \rangle^B]_{i=1}^{m}\gets \mathsf{packedCompare}([\lfloor \frac{m}{2} \rfloor-1]_m, [\langle s_i\rangle]_{i=1}^{m})$
        
    \State $[ q_i ]_{i=1}^{m} \gets \mathsf{Reveal}([\langle q_i \rangle^B]_{i=1}^{m})$, set $\mathcal{Y} = \{i : q_i = 1\}$

    \State $\langle \boldsymbol{g} \rangle \gets \mathsf{ADD}(\lbrace \mathsf{MUL}(\frac{|\mathcal{D}_i|}{\sum_i{|\mathcal{D}_i|}} \cdot \langle \boldsymbol{g}_i \rangle )\rbrace_{i \in \mathcal{Y}})$
    \State $\boldsymbol{g} \gets \mathsf{reveal}(\langle \boldsymbol{g} \rangle)$

    \State Broadcast $\boldsymbol{\omega}_{t+1} \gets \boldsymbol{\omega}_t - \boldsymbol{g}$ to $\mathcal{C}$
  \EndFor
  \State \Return $\boldsymbol{\omega}_{T}$
  \end{algorithmic}
\end{algorithm}

To prevent the leakage of neighborhood relationships among clients due to the revelation of comparison results during the quick select process~\cite{pbfl-li2023pbfl, shuffleandreveal-hamada2013practically}, the servers invoke $\mathsf{matrixSharedShuffle}$ to independently shuffle all rows of $\langle \boldsymbol{M} \rangle$ and produce $\langle \tilde{\boldsymbol{M}} \rangle$. The servers then use $\mathsf{mulRowQuickSelect}$ to compute shares of the medians for all rows, denoted $\lbrace \langle \mu_i \rangle \rbrace_{i \in [1,m]}$.

Recall that, half of $m$ clients closest to client $i$ are regarded as $i$'s neighborhood. That is, for any two clients $i, j \in [1, m]$, if $\langle {M}_{i,j} \rangle < \langle \mu_{i} \rangle$, client $j$ is a neighbor of client $i$, resulting in $\langle {N}_{i,j} \rangle = 1$ in the neighbor matrix $\langle \boldsymbol{N} \rangle$. 

The servers compute the sum of the $i$-th column of $\langle \boldsymbol{N} \rangle$ as the neighbor count $\langle s_i \rangle$ received by client $i$.

The servers compare the shares of $m$ pairs of neighbor counts and thresholds, where $\langle q_i \rangle^B = 1$ indicates that the neighbor count $s_i$ is greater than $\lfloor \frac{m}{2} \rfloor - 1$. In this case, $\boldsymbol{g}_i$ is considered benign and qualifies for aggregation in the current round.

Finally, the servers reveal $\langle q_i \rangle^B$ and aggregate the updates of the clients based on weights, resulting in a global update used to update the global model and proceed with the next round of global iteration. 

It is noteworthy that the global updates, obtained through splitting, aggregating, and revealing local model updates, exhibit virtually no impact on precision. The FLURP framework utilizes carefully designed secure and efficient SMPC protocols to enable servers to shuffle, anonymize, and collectively determine on the most reliable updates through a robust proximity-based defense. It also effectively diminishes the impact of potentially poisoned updates by ensuring that only those endorsed by a majority are considered.



\subsection{packedCompare}

The state-of-the-art Millionaires' protocol implementation, $\mathcal{F}_{\mathsf{Mill}}$~\cite{asscryptflow2-rathee2020cryptflow2}, eliminates the costly overhead associated with converting \textit{Arithmetic Sharing} to \textit{Yao Sharing}, which serve as inputs in garbled circuits~\cite{assABY-demmler2015aby}. We propose $\mathsf{packedCompare}$ (Algorithm~\ref{packedCompare}) that concatenates $n$ $l$-bit arithmetic shares bit-wise, enabling simultaneous comparison of multiple pairs with a round complexity of $\log l$, independent of the number of pairs. The variable $m$ denotes the bit length per segment, typically set to 4. The key steps are as follows:

Input preprocessing (lines 1-4): Each server $S_b$ computes the difference between corresponding shares and separates the most significant bit from the remaining bits. $S_0$ and $S_1$ then construct bit strings $\alpha$ and $\beta$ respectively, and parse them into $m$-bit segments and set $M=2^m$. 
Comparison tree initial (lines 5-7): The servers compute $[\langle 1 \{\alpha_j < \beta_j\} \rangle, \langle 1 \{\alpha_j = \beta_j\} \rangle]_{j=0}^{n \cdot q - 1}$ to construct $n \cdot q$ leaf nodes. We use $\langle lt_{0,j} \rangle_0^B$ and $\langle eq_{0,j} \rangle_0^B$ to represent $\langle 1 \{\alpha_j < \beta_j\} \rangle$ and $\langle 1 \{\alpha_j = \beta_j\} \rangle$, respectively, for each leaf node $j \in [0, n \cdot q - 1]$. For each node, $S_0$ first samples random bits $\langle lt_{0,j} \rangle_0^B, \langle eq_{0,j} \rangle_0^B$. For each value $k \in [0, M - 1]$, $S_0$ computes $s_{j,k} = \langle lt_{0,j} \rangle_0^B \oplus 1 \{\alpha_j < k\}$ and $t_{j,k} = \langle eq_{0,j} \rangle_0^B \oplus 1 \{\alpha_j = k\}$. $S_0$ and $S_1$ then invoke $\left( \substack{M \\ 1} \right) - \mathsf{OT}_2$ with $S_0$ inputting $[s_{j,k} || t_{j,k}]_{k=0}^{M-1}$ and $S_1$ inputting $\beta_j$. As a result, $S_1$ receives $\langle lt_{0,j} \rangle_1^B || \langle eq_{0,j} \rangle_1^B$ as output. 
Comparison tree iteration (lines 8-12): For each layer $i$ and node $j$ of comparison tree, $S_0$ and $S_1$ invoke $\mathcal{F}_{\mathsf{correlated AND}}$ with inputs $\langle lt_{i-1,2j} \rangle^B, \langle eq_{i-1,2j} \rangle^B, \langle lt_{i-1,2j+1} \rangle^B$. They receive outputs $\langle e \rangle^B$ and $\langle f \rangle^B$, where $e = lt_{i-1,2j} \land eq_{i-1,2j+1}$ and $f = eq_{i-1,2j} \land eq_{i-1,2j+1}$. Each server $S_b$ then sets $\langle lt_{i,j} \rangle_b^B = \langle lt_{i-1,2j} \rangle_b^B \oplus \langle e \rangle_b^B$ and $\langle eq_{i,j} \rangle_b^B = \langle f \rangle_b^B$. 
This process iteratively builds the comparison tree, combining results from lower layers to compute higher-layer outcomes. Final output calculation (lines 13-15): Each server $S_b$ computes the $n$ final comparison results in the $\log q$-th layer.

We compare $\mathsf{packedCompare}$ with $\mathcal{F}_{\mathsf{Mill}}$~\cite{asscryptflow2-rathee2020cryptflow2} over \textit{Arithmetic sharing} in Table~\ref{table:theoreticalcompare}. When $l = 32$, for comparing a single pair of arithmetic shares, both ideal functionalities $\mathcal{F}_{\mathsf{Mill}}$~\cite{asscryptflow2-rathee2020cryptflow2} and $\mathcal{F}_{\mathsf{packedCompare}}$ incur a communication overhead of 298 bits. For comparing $n$ pairs, $\mathcal{F}_{\mathsf{packedCompare}}$ achieves a round complexity of 5, significantly reducing the overall round complexity compared to the 32 rounds required by $\mathcal{F}_{\mathsf{Mill}}$.



\begin{algorithm}[!htbp]
\caption{$\mathsf{packedCompare}$} \label{packedCompare}
\begin{algorithmic}[1]
\Statex {\bf Input:} $S_0, S_1$ hold $\langle x_i \rangle, \langle y_i \rangle$ for $i \in [0, n-1]$
\Statex {\bf Output:} $S_0, S_1$ learn $\langle 1\{x_i < y_i\} \rangle^B$ for $i \in [0, n-1]$
\State $S_b$ computes $\langle x_i - y_i \rangle_b$ and parses as $\mathsf{msb}_{i,b} || w_{i,b}$ for $i \in [0, n-1], b\in[0,1]$
\State $S_0$ constructs $\alpha = (0 || 2^{l-1}-1-w_{0,0}) || \cdots || (0 || 2^{l-1}-1-w_{n-1,0})$
\State $S_1$ constructs $\beta = (0 || w_{0,1}) || \cdots || (0 || w_{n-1,1})$
\State Parse $\alpha$ and $\beta$ into $n \cdot q$ $m$-bit segments: $[\alpha_j]_{j=0}^{n \cdot q - 1}$ and $[\beta_j]_{j=0}^{n \cdot q - 1}$, where $q = \frac{l}{m}$
\For{$j \in [0, n \cdot q - 1]$}
\State Compute $\langle lt_{0,j} \rangle^B, \langle eq_{0,j} \rangle^B$ using $\left( \substack{M \\ 1} \right) - \mathsf{OT}_2$
\EndFor
\For{$i \in [1, \log q]$}
\For{$j \in [0, (n \cdot q /2^i) - 1]$}
\State Compute $\langle lt_{i,j} \rangle^B, \langle eq_{i,j} \rangle^B$ using $\mathcal{F}_{\mathsf{correlated AND}}$
\EndFor
\EndFor
\For{$j \in [0, n-1]$}
\State $S_b$ sets $\langle 1 \{x_j < y_j\} \rangle^B_b = \langle lt_{\log q,j} \rangle_b^B \oplus \mathsf{msb}_{j,b} \oplus b$
\EndFor
\end{algorithmic}
\end{algorithm}

\begin{table}[!htbp]
    \caption{The theoretical overhead of $\mathcal{F}_{\mathsf{packedCompare}}$ compared to two other comparison functionalities.}\label{table:theoreticalcompare}
    \begin{center} 
    \begin{adjustbox}{width=0.47\textwidth}
    \begin{tabular}{cccc} 
    \toprule
                                                                                                      & Protocol        & Communication(bits)                   & Round complexity  \\ 
    \hline
    \multirow{3}{*}{\begin{tabular}[c]{@{}c@{}}$1$ pair of $l$-bit \\arithmetic shares\end{tabular}}  & $\mathcal{F}_{\mathsf{Mill}}$~\cite{asscryptflow2-rathee2020cryptflow2}    & $\frac{l}{m}(2^{m+1}+6)-6$    & $\boldsymbol{logl}$         \\
                                                                                                      & $\mathcal{F}_{\mathsf{packedCompare}}$ & $\frac{l}{m}(2^{m+1}+6)-6$    & $\boldsymbol{logl}$         \\
    \hline
    \multirow{3}{*}{\begin{tabular}[c]{@{}c@{}}$n$ pairs of $l$-bit \\arithmetic shares\end{tabular}} & $\mathcal{F}_{\mathsf{Mill}}$~\cite{asscryptflow2-rathee2020cryptflow2}    & $n(\frac{l}{m}(2^{m+1}+6)-6)$ & $nlogl$       \\
                                                                                                      & $\mathcal{F}_{\mathsf{packedCompare}}$ & $n(\frac{l}{m}(2^{m+1}+6)-6)$ & $\boldsymbol{logl}$         \\
    \bottomrule
    \end{tabular}
    \end{adjustbox}
    \end{center}
    \end{table}


\subsection{matrixSharedShuffle}

The $\mathsf{matrixSharedShuffle}$ (Algorithm~\ref{matrixSharedShuffle}) independently shuffles each row of the shared matrix $\langle \boldsymbol{D} \rangle$ in three rounds, with a communication overhead of only $4m^2$ AHE ciphertexts. This protocol serves as a prerequisite for the subsequent $\mathsf{mulRowPartition}$ and $\mathsf{mulRowQuickSelect}$ protocols, designed to prevent the leakage of client indices during the comparison of the distances between revealed updates.

First, $S_1$ encrypts the elements of $\langle \boldsymbol{D} \rangle_1$ using the AHE public key $pk_1$, resulting in $\llbracket \langle \boldsymbol{D} \rangle_{1} \rrbracket_1$, a collection of $m^2$ AHE ciphertexts, which is then sent to $S_0$. Next, $S_0$ decrypts $\llbracket \langle \boldsymbol{D} \rangle_{1} \rrbracket_1$ to obtain $\boldsymbol{D}$ in plaintext form and applies a randomly generated mask matrix $\boldsymbol{L}$. $S_0$ then shuffles the matrices $\llbracket \boldsymbol{L} \rrbracket_0$ and $\llbracket \boldsymbol{D} - \boldsymbol{L} \rrbracket_1$ using the permutation $\boldsymbol{\pi}_0$ via $\mathsf{matrixShuffle}$. Specifically, $\mathsf{matrixShuffle}(\boldsymbol{M}, \boldsymbol{\pi})$ shuffles each row of the input matrix $\boldsymbol{M}$ with $\boldsymbol{M}_{i,j} \gets \boldsymbol{M}_{i, {\pi}_{i,j}}$, for $i,j \in [0, m-1]$. The shuffled matrices $\llbracket \boldsymbol{L}^{\boldsymbol{\pi}_0} \rrbracket_0$ and $\llbracket (\boldsymbol{D} - \boldsymbol{L})^{\boldsymbol{\pi}_0} \rrbracket_1$, with a total size of $2m^2$ AHE ciphertexts, are then sent to $S_1$.

Subsequently, $S_1$ decrypts to obtain the masked matrix $(\boldsymbol{D} - \boldsymbol{L})^{\boldsymbol{\pi}_0}$ in its permuted plaintext form. Since $\boldsymbol{L}$ contains randomly selected elements, $S_1$ gains no information about $\boldsymbol{D}$. Then, $S_1$ applies a new randomly sampled mask matrix $\boldsymbol{R}$ to $(\boldsymbol{D} - \boldsymbol{L})^{\boldsymbol{\pi}_0}$ and permutes it with $\boldsymbol{\pi}_1$, resulting in $\langle \tilde{\boldsymbol{D}} \rangle_1$, which equals $((\boldsymbol{D} - \boldsymbol{L})^{\boldsymbol{\pi}_0} - \boldsymbol{R})^{\boldsymbol{\pi}_1}$. Simultaneously, $S_1$ applies the mask $\boldsymbol{R}$ to the plaintext matrix $\boldsymbol{L}^{\boldsymbol{\pi}_0}$ under the ciphertext $\llbracket \boldsymbol{L}^{\boldsymbol{\pi}_0} \rrbracket_0$. The result $\llbracket \boldsymbol{L}^{\boldsymbol{\pi}_0} + \boldsymbol{R} \rrbracket_0$ is then permuted by $\boldsymbol{\pi}_1$ and sent to $S_0$. Finally, $S_0$ decrypts $\llbracket (\boldsymbol{L}^{\boldsymbol{\pi}_0} + \boldsymbol{R})^{\boldsymbol{\pi}_1} \rrbracket_0$ to obtain the plaintext $(\boldsymbol{L}^{\boldsymbol{\pi}_0} + \boldsymbol{R})^{\boldsymbol{\pi}_1}$, setting this as $\langle \tilde{\boldsymbol{D}} \rangle_0$.

\begin{algorithm}[!htbp]
  \caption{$\mathsf{matrixSharedShuffle}$}\label{matrixSharedShuffle}
  \begin{algorithmic}[1]
  \Statex {\bf Input:} $S_0$ holds $\langle \boldsymbol{ D } \rangle_0$ and $\boldsymbol{\pi}_0$,  $S_1$ holds $\langle \boldsymbol{ D } \rangle_1$ and $\boldsymbol{\pi}_1$, where $\langle \boldsymbol{ D } \rangle = [\langle D_{i,j} \rangle]_{i,j=1}^{m,n}$
  \Statex {\bf Output:} $S_0, S_1$ learn $\langle \tilde{\boldsymbol{D}} \rangle$ satisfying $\langle \tilde{\boldsymbol{D}} \rangle = \langle \boldsymbol{D}^{\boldsymbol{\pi}_0, \boldsymbol{\pi}_1} \rangle$
    \State $\triangleright S_1$:
        \State $\llbracket \langle \boldsymbol{D} \rangle_1 \rrbracket_1 \gets \mathsf{Enc}(\langle \boldsymbol{D} \rangle_1, pk_1)$
        \State Send $\llbracket \langle \boldsymbol{D} \rangle_1 \rrbracket_1$ to $S_0$
    \State $\triangleright S_0$:
        \State $\llbracket \boldsymbol{D} \rrbracket_1 \gets \llbracket \langle \boldsymbol{D} \rangle_1 \rrbracket_1 \boxplus \langle \boldsymbol{D} \rangle_0$
        \State $\boldsymbol{L} \gets [L_{i,j} \in_R  \{0,1\}^{\kappa}]_{i,j=1}^{m,n}$
        \State $\llbracket \boldsymbol{D} - \boldsymbol{L} \rrbracket_1 \gets \llbracket \boldsymbol{D} \rrbracket_1 \boxplus (-\boldsymbol{L})$
        \State $\llbracket (\boldsymbol{D} - \boldsymbol{L})^{\boldsymbol{\pi}_0} \rrbracket_1 \gets \mathsf{matrixShuffle}(\llbracket \boldsymbol{D} - \boldsymbol{L} \rrbracket_1, \boldsymbol{\pi}_0)$
        \State $\llbracket \boldsymbol{L}^{\boldsymbol{\pi}_0} \rrbracket_0 \gets \mathsf{matrixShuffle}(\llbracket \boldsymbol{L} \rrbracket_0, \boldsymbol{\pi}_0)$
        \State Send $\llbracket (\boldsymbol{D} - \boldsymbol{L})^{\boldsymbol{\pi}_0} \rrbracket_1, \llbracket \boldsymbol{L}^{\boldsymbol{\pi}_0} \rrbracket_0$ to $S_1$
    \State $\triangleright S_1$:
        \State $(\boldsymbol{D} - \boldsymbol{L})^{\boldsymbol{\pi}_0} \gets \mathsf{Dec}(\llbracket (\boldsymbol{D} - \boldsymbol{L})^{\boldsymbol{\pi}_0} \rrbracket_1, pk_1)$
        \State $\boldsymbol{R} \gets [R_{i,j} \in_R  \{0,1\}^{\kappa}]_{i,j=1}^{m,n}$
        \State $\llbracket \boldsymbol{L}^{\boldsymbol{\pi}_0}+\boldsymbol{R} \rrbracket_0 \gets \llbracket \boldsymbol{L}^{\boldsymbol{\pi}_0} \rrbracket_0 \boxplus \boldsymbol{R}$
        \State $\langle \tilde{\boldsymbol{D}} \rangle_1 \gets \mathsf{matrixShuffle}((\boldsymbol{D} - \boldsymbol{L})^{\boldsymbol{\pi}_0} - \boldsymbol{R}, \boldsymbol{\pi}_1)$
        \State $\llbracket (\boldsymbol{L}^{\boldsymbol{\pi}_0}+\boldsymbol{R})^{\boldsymbol{\pi}_1} \rrbracket_0 \gets \mathsf{matrixShuffle}(\llbracket \boldsymbol{L}^{\boldsymbol{\pi}_0}+\boldsymbol{R} \rrbracket_0, \boldsymbol{\pi}_1)$
        \State Send $\llbracket (\boldsymbol{L}^{\boldsymbol{\pi}_0}+\boldsymbol{R})^{\boldsymbol{\pi}_1} \rrbracket_0$ to $S_0$
    \State $\triangleright S_0$:
        \State $\langle \tilde{\boldsymbol{D}} \rangle_0 \gets \mathsf{Dec}(\llbracket (\boldsymbol{L}^{\boldsymbol{\pi}_0}+\boldsymbol{R})^{\boldsymbol{\pi}_1} \rrbracket_0, pk_0)$
  \end{algorithmic}
\end{algorithm}



    
  
  

    

\subsection{mulRowPartition}

$\mathsf{mulRowPartition}$ (Algorithm~\ref{mulRowPartition}) is a fundamental component of $\mathsf{mulRowQuickSelect}$. It takes $m$ randomly shuffled shared sequences as input and produces $m$ partitioned sequences along with indices for $m$ pivots.

The protocol $\mathsf{mulRowPartition}$ partitions each sequence by using its last element as a pivot. Line 1 packs $m$ pivots into a set. In lines 4-7, the algorithm appends the sharing of the $j$-th position of $\boldsymbol{A}_i$ and its corresponding pivot to $\alpha$ and $\beta$ respectively. Meanwhile, $\boldsymbol{d}$ records the indices of the rows involved in comparison. In lines 8-16, after revealing the comparison results, the value $\boldsymbol{A}_{d_i,j}$ smaller than the pivot of row $\boldsymbol{A}_{d_i}$ will be swapped to the $t_{d_i}$-th position of $\boldsymbol{A}_{d_i}$. Upon completion, each sequence $\langle \boldsymbol{A}_i \rangle$ is partitioned such that all shared values to the left of $\langle \boldsymbol{A}_{{q}_i}\rangle$ are smaller than it, and those to the right are larger. The input shared sequences may have different lengths, with the complexity of invoking $\mathsf{packedCompare}$ corresponds to the length of the longest sequence.


\begin{algorithm}[!htbp]
  \caption{$\mathsf{mulRowPartition}$}\label{mulRowPartition}
  \begin{algorithmic}[1]
  \Statex {\bf Input:} $S_0, S_1$ hold a shared collection $\langle \boldsymbol{A} \rangle = [\langle \boldsymbol{A}_i \rangle]_{i=0}^{m-1}$
  \Statex {\bf Output:} $S_0, S_1$ learn the index set $\boldsymbol{q} = [q_i]_{i=0}^{m-1}$, and a shared collection $\langle \boldsymbol{A} \rangle = [\langle \boldsymbol{A}_i \rangle]_{i=0}^{m-1}$
    \State $\boldsymbol{p}, \boldsymbol{t} \gets [\langle A_{i,-1} \rangle]_{i=0}^{m-1}, [-1]_m$
    \State $l \gets \max(\{\mathsf{len}(\langle \boldsymbol{A}_i \rangle)\}_{i=0}^{m-1})$
    \For{$j \in [0, l-2]$}
        \State $\boldsymbol{\alpha}, \boldsymbol{\beta}, \boldsymbol{d} \gets \emptyset$
        \For{$i \in [0,m-1]$ \textbf{where} $j < \mathsf{len}(\boldsymbol{A}_i) - 1$}
              \State Append $\langle A_{i,j} \rangle, \langle p_{i} \rangle, i$ to $\boldsymbol{\alpha}, \boldsymbol{\beta}, \boldsymbol{d}$
        \EndFor
        \State $\boldsymbol{c} \gets \mathsf{Reveal}(\mathsf{packedCompare}(\boldsymbol{\alpha}, \boldsymbol{\beta}))$
        \For{$i \in [0, \mathsf{len}(\boldsymbol{c})-1]$ \textbf{where} $c_i == 1$}
              \State $t_{d_i} \gets t_{d_i} + 1$
              \State Swap $\langle A_{d_i,t_{d_i}} \rangle$ and $\langle A_{d_i,j} \rangle$
        \EndFor
        \State $\boldsymbol{q} \gets [t_i + 1]_{i=0}^{m-1}$
        \For{$i \in [0, m-1]$}
            \State Swap $\langle A_{i, q_i} \rangle$ and $\langle A_{i,-1} \rangle$
        \EndFor     
    \EndFor
  \State \Return $\boldsymbol{q}, \langle \boldsymbol{A} \rangle$
  \end{algorithmic}
\end{algorithm}


\subsection{mulRowQuickSelect}

The protocol $\mathsf{mulRowQuickSelect}$ (Algorithm~\ref{mulRowQuickSelect}) takes as input a collection of $m$ randomly and independently shuffled shared sequences $\langle \boldsymbol{A} \rangle$, a target set $\boldsymbol{t} = [t_i]_{i=0}^{m-1}$, and a source index set $\boldsymbol{x} = [x_i]_{i=0}^{m-1}$. where each sequence $\langle \boldsymbol{A}_i \rangle$ may have a different length. After recursive calls to $\mathsf{mulRowQuickSelect}$, the keys of $\bar{\boldsymbol{d}}$ are the row indices, and the values are the $t_i$-th largest elements of $\langle \boldsymbol{A}_i \rangle$.

In the initial global call to $\mathsf{mulRowQuickSelect}$, $\bar{\boldsymbol{d}}$ is an empty dictionary, and $x_i = i$.
 
First, $\mathsf{mulRowQuickSelect}$ calls $\mathsf{mulRowPartition}$ to partition all rows. In loop of lines 4-13, for each partitioned row $\langle \boldsymbol{A}_i \rangle$, $\langle \boldsymbol{L_i} \rangle$ stores the shared values smaller than the pivot, and $\langle \boldsymbol{R}_i \rangle$ stores the pivot and those shared numbers greater than it. Next, if the length of $\langle \boldsymbol{R}_i \rangle$ matches $t_i$ exactly, this pivot is recorded as a value in $\bar{\boldsymbol{d}}$ with the corresponding key as $x_i$. Otherwise, the sequence containing the target value is added to the collection $\langle \boldsymbol{A'} \rangle$, and the new target $\boldsymbol{t'}$, along with their initial row numbers $\boldsymbol{x'}$, are computed and used as inputs to recursively call $\mathsf{mulRowQuickSelect}$. Therefore, with each call to $\mathsf{mulRowQuickSelect}$, the algorithm adds at least one value to the global dictionary $\bar{\boldsymbol{d}}$. Since $\langle \boldsymbol{A}_0 \rangle, \ldots \langle \boldsymbol{A}_{m-1} \rangle$ are all randomly shuffled. The complexity of calling $\mathsf{packedCompare}$ is equal to the length of the longest sequence.


\begin{algorithm}[!htbp]
  \caption{$\mathsf{mulRowQuickSelect}$}\label{mulRowQuickSelect}
  \begin{algorithmic}[1]
  \Statex {\bf Input:} $S_0, S_1$ hold $\langle \boldsymbol{A} \rangle = [\langle \boldsymbol{A}_i \rangle]_{i=0}^{m-1}$, $\boldsymbol{t} = [t_i]_{i=0}^{m-1}$, $\boldsymbol{x} = [x_i]_{i=0}^{m-1}$
  \Statex {\bf Output:} $S_0, S_1$ learn dictionary $\langle \bar{\boldsymbol{d}} \rangle = [\langle \bar{\boldsymbol{d}}_i \rangle]_{i=0}^{m-1}$, where $\bar{\boldsymbol{d}}_i$ is $t_i\text{-th}$ largest of $\boldsymbol{A}_i$
    \State \textbf{global dictionary} $\bar{\boldsymbol{d}}$
    \State $\boldsymbol{p}, \langle \boldsymbol{A} \rangle \gets \mathsf{mulRowPartition}(\langle \boldsymbol{A} \rangle)$
    \State Initialize $\langle \boldsymbol{R} \rangle, \langle \boldsymbol{L} \rangle, \langle \boldsymbol{A'} \rangle, \boldsymbol{t'}, \boldsymbol{x'}$ as empty lists
    
    \For{$i \in [0, m-1]$}
        \State $\langle \boldsymbol{R}_i \rangle, \langle \boldsymbol{L}_i \rangle \gets \mathsf{slice}(\langle \boldsymbol{A}_i \rangle, p_i, -1), \mathsf{slice}(\langle \boldsymbol{A}_i \rangle, 0, p_i-1)$
        \If{$\mathsf{len}(\langle \boldsymbol{R}_i \rangle) == t_i$}
            \State $\langle \bar{\boldsymbol{d}}_{x_i} \rangle \gets \langle A_{i,p_i} \rangle$
        \ElsIf{$\mathsf{len}(\langle \boldsymbol{R}_i \rangle) > t_i$}
            \State Append $\langle \boldsymbol{R}_i \rangle, t_i, x_i$ to $\langle \boldsymbol{A'} \rangle, \boldsymbol{t'}, \boldsymbol{x'}$
        \Else
            \State Append $\langle \boldsymbol{L}_i \rangle, t_i - \mathsf{len}(\langle \boldsymbol{R}_i \rangle), x_i$ to $\langle \boldsymbol{A'} \rangle, \boldsymbol{t'}, \boldsymbol{x'}$
        \EndIf
    \EndFor
    \If{$\boldsymbol{x'} \neq \emptyset$}
        \State $\mathsf{mulRowQuickSelect}(\langle \boldsymbol{A'} \rangle, \boldsymbol{t'}, \boldsymbol{x'})$
    \EndIf
\end{algorithmic}
\end{algorithm}



%
\section{Security Analysis}

\begin{theorem}\label{theorem:packedCompare}
    The security of $\mathsf{packedCompare}$ holds in the $(\left( \substack{M \\ 1} \right) - \mathsf{OT}_2, \mathcal{F}_{\mathsf{correlated AND}})$-hybrid, as $\lbrace\langle lt_{0,j} \rangle^B, \langle eq_{0,j} \rangle^B \rbrace_{j \in [0,n \cdot q -1]}$ are uniformly random.
\end{theorem}

\begin{proof}

    The security of $\mathsf{packedCompare}$ relies on the secure comparison protocol $\mathcal{F}_{\mathsf{Mill}}$ \cite{asscryptflow2-rathee2020cryptflow2}, which compares $l$-bit unsigned integers by parsing them into $m$-bit segments. These segments are treated as leaf nodes in a comparison tree. We extend this approach by concatenating $n$ $l$-bit unsigned integers at the leaf node level (lines 2-3 of Algorithm \ref{packedCompare}), maintaining the tree structure and number of layers ($\log q$). Since $\mathsf{packedCompare}$ inherits the security of $\mathcal{F}_{\mathsf{Mill}}$, the security of $\mathsf{packedCompare}$ is preserved.
\end{proof}

\begin{theorem}\label{theorem:matrixSharedShuffle}

   The protocol $\mathsf{matrixSharedShuffle}$ securely shuffles the shared matrix row-wise under the threat of an honest-but-curious adversary $\mathcal{A}_h$, provided that AHE is semantically secure.
\end{theorem}

\begin{proof}

Using the ideal/real-world paradigm with functionality $\mathcal{F}_{\mathsf{matrixSharedShuffle}}$, we analyze the security when adversary $\mathcal{A}_h$ compromises $S_0$ or $S_1$, aiming for computational indistinguishability.

\textit{Case 1 (Compromised $S_0$)}:  
Simulator $\mathsf{Sim}$:
\begin{itemize}
    \item Receives $\langle \boldsymbol{D} \rangle_{0}$, $\boldsymbol{\pi}_0$, $\boldsymbol{L}$ from $\mathsf{Env}$, forwards them to $\mathcal{F}_{\mathsf{matrixSharedShuffle}}$ to get $\boldsymbol{w}$, $\boldsymbol{z}$.
    \item Constructs $\boldsymbol{w}' = \mathsf{Enc}(\boldsymbol{0}, pk_1)$ and $\boldsymbol{z}' = \mathsf{Enc}(\boldsymbol{z}, pk_1)$; sends them to $S_0$.
    \item Outputs everything $S_0$ outputs.
\end{itemize}

The simulated view for $S_0$ is indistinguishable from the real view due to the semantic security of AHE. Even with $\boldsymbol{L}^{\boldsymbol{\pi}_0}$ known, $S_0$ cannot deduce $\boldsymbol{\pi}_1$ from $\boldsymbol{z}$ due to randomness in $\boldsymbol{R}$.

\textit{Case 2 (Compromised $S_1$)}:  
Simulator $\mathsf{Sim}$:
\begin{itemize}
    \item Receives $\langle \boldsymbol{D} \rangle_{1}$, $\boldsymbol{\pi}_1$, $\boldsymbol{R}$ from $\mathsf{Env}$, forwards them to $\mathcal{F}_{\mathsf{matrixSharedShuffle}}$ to get $\boldsymbol{x}$, $\boldsymbol{y}$.
    \item Constructs $\boldsymbol{x}' = \mathsf{Enc}(\boldsymbol{x}, pk_1)$ and $\boldsymbol{y}' = \mathsf{Enc}(\boldsymbol{0}, pk_0)$; sends them to $S_1$.
    \item Outputs everything $S_1$ outputs.
\end{itemize}

The simulated view for $S_1$ is indistinguishable from the real view because $\boldsymbol{y}' = \llbracket 0 \rrbracket_0$ is indistinguishable from $\boldsymbol{y}^{\text{real}} = \llbracket \langle \boldsymbol{D} \rangle_0 \rrbracket_0$. Even with $sk_1$, $\boldsymbol{x} = (\boldsymbol{D} - \boldsymbol{L})^{\boldsymbol{\pi}_0}$ reveals nothing about $\boldsymbol{\pi}_0$ due to randomness in $\boldsymbol{L}$.

In both cases, the ideal-world output distribution is computationally indistinguishable from the real-world output, completing the proof.

\end{proof}

\begin{theorem}\label{theorem:mulRowPartition}

    The security of $\mathsf{mulRowPartition}$ holds in the $(\mathcal{F}_{\mathsf{packedCompare}}, \mathcal{F}_{\mathsf{matrixSharedShuffle}})$-hybrid.
    
\end{theorem}

\begin{proof}

We define the following hybrids to show indistinguishability between the simulated and real-world views of $\mathsf{mulRowPartition}$ when $S_1$ is compromised.
\textbf{Hybrid $\mathcal{H}_0$}: Real-world execution.

\textbf{Hybrid $\mathcal{H}_1$}: Same as $\mathcal{H}_0$, except $S_1$ uses a random permutation $\boldsymbol{\pi}_1'$ independent of $\boldsymbol{\pi}_0$ and invokes $\mathcal{F}_{\mathsf{matrixSharedShuffle}}$ to obtain $\langle \tilde{\boldsymbol{D}}' \rangle = \langle \boldsymbol{D}^{\boldsymbol{\pi}_0, \boldsymbol{\pi}_1'} \rangle$.

\textbf{Hybrid $\mathcal{H}_2$}: Same as $\mathcal{H}_1$, except during $l-1$ iterations, $S_0$ and $S_1$ invoke $\mathcal{F}_{\mathsf{packedCompare}}$ to obtain $\langle \boldsymbol{c}' \rangle$. $S_1$ sends $\langle \boldsymbol{c}' \rangle_1$ to $S_0$, which adjusts $\langle \tilde{\boldsymbol{D}}' \rangle$ and $\boldsymbol{q}'$ based on $\boldsymbol{c}'$.

\textbf{Hybrid $\mathcal{H}_3$}: Ideal-world execution. Simulator $\mathsf{Sim}$ sends random $\boldsymbol{q}''$ and $\langle \tilde{\boldsymbol{D}}'' \rangle$ to $S_0$.

\textbf{Analysis}: 
\begin{itemize}
    \item \emph{Indistinguishability of $\mathcal{H}_0$ and $\mathcal{H}_1$}: By Theorem \ref{theorem:matrixSharedShuffle}, $\langle \tilde{\boldsymbol{D}}' \rangle_0$ and $\langle \tilde{\boldsymbol{D}}' \rangle_1$ are indistinguishable from the real-world view.
    \item \emph{Indistinguishability of $\mathcal{H}_1$ and $\mathcal{H}_2$}: By Theorem \ref{theorem:packedCompare}, $\mathcal{F}_{\mathsf{packedCompare}}$ does not leak original values. Random permutations ensure that $S_0$ cannot deduce client indices from $\langle \boldsymbol{c}' \rangle$.
    \item \emph{Indistinguishability of $\mathcal{H}_2$ and $\mathcal{H}_3$}: In $\mathcal{H}_2$, $\boldsymbol{q}'$ and $\langle \tilde{\boldsymbol{D}}' \rangle$ are independent of the original values. Thus, they are indistinguishable from the random values in $\mathcal{H}_3$.
\end{itemize}
By transitivity, the simulated view is indistinguishable from the real-world view, completing the proof.

\end{proof}

\begin{theorem}\label{theorem:mulRowQuickSelect}
    The security of $\mathsf{mulRowQuickSelect}$ holds trivially in the $(\mathcal{F}_{\mathsf{packedCompare}}, \mathcal{F}_{\mathsf{mulRowPartition}})$-hybrid.
\end{theorem}

\begin{proof}
We use a hybrid argument to prove that the simulated view of $\mathsf{mulRowQuickSelect}$ is indistinguishable from the real-world view.

\textbf{Hybrid $\mathcal{H}_0$}: Real-world execution of $\mathsf{mulRowQuickSelect}$.


\textbf{Hybrid $\mathcal{H}_1$}: Identical to $\mathcal{H}_0$, except $S_0$ inputs $\langle \boldsymbol{D} \rangle_{0}$ and $\boldsymbol{\pi}_0$, and $S_1$ inputs $\langle \boldsymbol{D} \rangle_{1}$ and a randomly generated $\boldsymbol{\pi}_1'$ independent of $\boldsymbol{\pi}_0$.
Both parties invoke $\mathcal{F}_{\mathsf{matrixSharedShuffle}}$, and the output $\langle \tilde{\boldsymbol{D}}' \rangle = \langle \boldsymbol{D}^{\boldsymbol{\pi}_0, \boldsymbol{\pi}_1'} \rangle$ is added to the simulated view.

\textbf{Hybrid $\mathcal{H}_2$}: Same as $\mathcal{H}_1$, except $S_0$ and $S_1$ invoke $\mathcal{F}_{\mathsf{mulRowPartition}}$ in up to $m-1$ iterations.

\textbf{Hybrid $\mathcal{H}_3$}: Ideal-world execution. Simulator $\mathsf{Sim}$ sends random $\langle \bar{\boldsymbol{d}}'' \rangle_0$ to $S_0$.

\textbf{Analysis}: 
\begin{itemize}
    \item \emph{Indistinguishability of $\mathcal{H}_0$ and $\mathcal{H}_1$}: By Theorem \ref{theorem:matrixSharedShuffle}, the outputs in $\mathcal{H}_1$ are indistinguishable from the real-world view.
    \item \emph{Indistinguishability of $\mathcal{H}_1$ and $\mathcal{H}_2$}: By Theorem \ref{theorem:mulRowPartition}, the outputs in $\mathcal{H}_2$ are indistinguishable from those in $\mathcal{H}_1$.
    \item \emph{Indistinguishability of $\mathcal{H}_2$ and $\mathcal{H}_3$}: In $\mathcal{H}_2$, the output $\langle \bar{\boldsymbol{d}}' \rangle_0$ is uniformly random and independent of the real-world inputs. Thus, it matches the random output in $\mathcal{H}_3$.
\end{itemize}
By transitivity, the simulated view is indistinguishable from the real-world view, completing the proof.
\end{proof}

\section{Evaluation}

\subsection{Experiment Settings}\label{sec:modelsetting}

\textbf{Environment.} The experiments are conducted on Ubuntu 20.04 with an AMD 3960X 24-core CPU, an NVIDIA GTX 3090 24GB GPU, and 64GB RAM. We use PyTorch for training, the Python PHE library for AHE implementation, and C++ for SMPC protocols involving secret sharing.

\textbf{Tasks.} We perform the following tasks: (1) Train on 50,000 images and test on 10,000 from \textbf{CIFAR-10} \cite{cifar}, which contains 60,000 $32 \times 32$ color images across 10 categories, using the \textbf{ResNet10} model with 4.9M parameters. (2) Train on 12,480 images and test on 3,120 from \textbf{ImageNet-12} \cite{li2023reconstructive}, a subset of ImageNet~\cite{deng2009imagenet} with $224 \times 224$ color images across 12 classes, using the \textbf{MobileNet-V2}~\cite{sandler2018mobilenetv2} model with 2.2M parameters. (3) Evaluate \textbf{AgNews} \cite{zhang2015character}, containing 120,000 training and 7,600 testing samples of news articles in four categories (World, Sports, Business, Science/Technology), using a \textbf{Bi-LSTM} model with 4.6M parameters, pretrained BERT tokenizer, 128-dimensional embeddings, 256 hidden units per direction, and a 4-class output layer for classification.

\textbf{Clients.} We configure training parameters in Table~\ref{tab:trainingpara}. Forty percent of clients are malicious, capable of performing one of the eight attacks described in Section~\ref{sec:byzantine}. The training set is divided under both IID and non-IID conditions.

\begin{table}[h]
\centering
\caption{Training configuration for three tasks}\label{tab:trainingpara}
\begin{threeparttable}
\begin{tabular}{@{}lccc@{}}
\toprule
Parameter            & Task 1 & Task 2 & Task 3 \\ \midrule
\# of clients    & 20              & 10              & 20              \\
Learning rate  & 0.1             & 0.1             & 0.1             \\
Optimizer            & SGD ($t=0.9$)     & SGD ($t=0.9$)     & SGD ($t=0.9$)     \\
Batch size           & 128             & 128             & 128             \\
Local epochs         & 10              & 3               & 1               \\
Global rounds        & 200             & 50              & 30              \\ \bottomrule
\end{tabular}
\begin{tablenotes}
    \footnotesize
    \item Note: $t$ denotes the momentum value used in the SGD optimizer.
\end{tablenotes}
\end{threeparttable}
\end{table}

\textbf{Evaluation metrics.} Let $\mathcal{D}_g$ denote the global test set. For untargeted attacks, we use Main Accuracy (MA), equivalent to Top-1 Accuracy, evaluated on the global model using $\mathcal{D}_g$. For targeted attacks, we employ Attack Success Rate (ASR) using test data $\mathcal{D}_t$, created by removing target label samples from $\mathcal{D}_g$ and inserting triggers into the remaining samples. ASR is the proportion of $\mathcal{D}_t$ samples classified as the target class by the global model. Higher MA and lower ASR indicate better defense effectiveness.

\subsection{Comparison of Sampling Methods} \label{sec:comparesampling}

\begin{figure}[!htbp]
\centering
\includegraphics[width=0.48\textwidth]{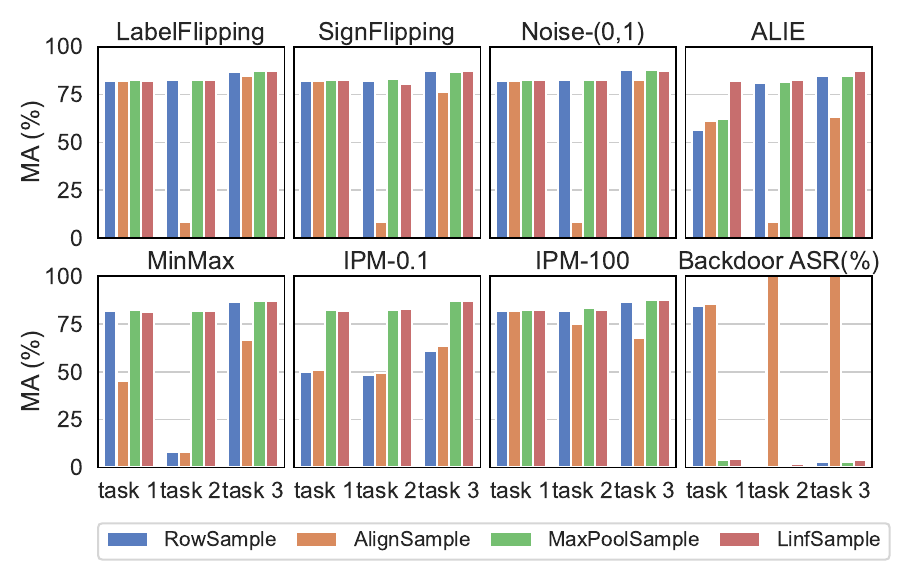}
\caption{The MA and ASR (\%) of the three baselines and $\mathsf{LinfSample}$ across three training tasks.}
\label{fig:sample-methods-resluts}
\end{figure}

To demonstrate the superior resistance of $\mathsf{LinfSample}$ against Byzantine attacks and its ability to reduce SMPC overhead in computing the shared SED matrix, we apply the same proximity-based defense to evaluate four different sampling methods. The configurations for these methods are as follows:

(i) $\mathsf{RowSample}$: No sampling is performed. Servers compute the shared SED matrix on updates. (ii) $\mathsf{AlignSample}$~\cite{RFBDS-chen2023privacy}: For each layer of the update, the maximum absolute value is determined first, and all entries are aligned with this value while retaining their original signs. (iii) $\mathsf{MaxPoolSample}$~\cite{agramplifier-gong2023agramplifier}: The update is reshaped into a matrix and subjected to 2D max-pooling with a kernel size of $5 \times 5$. (iv) $\mathsf{LinfSample}$: The update is computed as described in Section~\ref{subsec:linfsample}, with the window size set to $4096$. In (ii), (iii), and (iv), servers compute the shared SED matrix based on the sampling results. 

\subsubsection{Byzantine resilience} The results in Figure~\ref{fig:sample-methods-resluts} show that for $\mathsf{RowSample}$, due to the fact that SED is computed based on the full-size updates, its performance is excellent under most attacks. However, it fails in Task 1 of ALIE and Task 2 of MinMax. $\mathsf{AlignSample}$ performs poorly in most cases, with only slight improvement in defense under the IMP-100 attack, which amplifies updates on a large scale. Notably, the IPM-0.1 attack reduces the global model accuracy for both $\mathsf{RowSample}$ and $\mathsf{AlignSample}$. $\mathsf{MaxPoolSample}$ performs slightly worse than $\mathsf{LinfSample}$ only in Task 1 of ALIE. The $\mathsf{LinfSample}$ method we designed performs optimally against most attacks, except for a slight inferiority to $\mathsf{MaxPoolSample}$ in Task 2 of the SignFlipping attack.

\subsubsection{Communication of client uploads for Sharing}\label{commofclient}

The $\mathsf{RowSample}$ method requires only the splitting of updates, resulting in the lowest overhead for client-side sharing. In contrast, the $\mathsf{LinfSample}$, $\mathsf{AlignSample}$, and $\mathsf{MaxPoolSample}$ methods involve splitting both updates and the sampling results. This additional step increases the overhead, which depends on the size of the sampling results. As shown in Table~\ref{table:LinfSample}, the sampling results for $\mathsf{LinfSample}$ are the smallest among the three methods, making its additional overhead negligible by comparison.

\subsubsection{Overhead of computing shared SED matrix}

\begin{table}[!htbp]
\centering
\caption{Overhead comparison of $\mathsf{LinfSample}$ and three other sampling methods per round.}
\label{table:LinfSample}
\begin{adjustbox}{width=0.47\textwidth}
\begin{tabular}{ccccc} 
\toprule
\multirow{2}{*}{Dataset} & $\mathsf{RowSample}$ & $\mathsf{AlignSample}$ & $\mathsf{MaxPoolSample}$ & $\mathsf{LinfSample}$  \\ 
\cline{2-5}
                         & \multicolumn{4}{c}{Communication of sharing for each client (MB)}                                 \\ 
\hline
CIFAR-10                 & \textbf{18.8}        & 37.5                   & 19.5                     & 18.8                   \\
ImageNet-12              & \textbf{5.7}         & 11.3                   & 5.9                      & 5.7                    \\
AgNews                   & \textbf{0.6}         & 1.1                    & 0.6                      & 0.6                    \\ 
\hline
                         & \multicolumn{4}{c}{Communication of computing shared SED matrix for servers (MB)}                 \\ 
\hline
CIFAR-10                 & 14215.3              & 14215.3                & 568.6                    & \textbf{3.5}           \\
ImageNet-12              & 1561.0               & 1561.0                 & 62.4                     & \textbf{0.4}           \\
AgNews                   & 13624.3              & 13624.3                & 545.0                    & \textbf{3.3}           \\ 
\hline
                         & \multicolumn{4}{c}{Runtime of computing shared SED matrix for servers (s)}                        \\ 
\hline
CIFAR-10                 & 726.5                & 726.5                  & 30.2                     & \textbf{0.4}           \\
ImageNet-12              & ~80.1                & 80.1                   & 3.3                      & \textbf{0.1}           \\
AgNews                   & 685.3                & 685.3                  & 27.8                     & \textbf{0.4}           \\
\bottomrule
\end{tabular}
\end{adjustbox}
\end{table}

We further evaluate the communication and time overhead of computing the SED matrix for the four sampling methods, as shown in Table~\ref{table:LinfSample}. Since all four methods compute the same number of SEDs, the overhead is primarily determined by the length of the input shared vectors. Specifically, the length of $\langle \boldsymbol{g}_i \rangle$ matches that of the sampling result in $\mathsf{AlignSample}$, is 25 times that in $\mathsf{MaxPoolSample}$, and 4096 times that of $\langle \boldsymbol{v}_i \rangle$ in $\mathsf{LinfSample}$. The results indicate that $\mathsf{LinfSample}$ reduces both communication and runtime by three orders of magnitude compared to $\mathsf{RowSample}$ and $\mathsf{AlignSample}$. Furthermore, compared to $\mathsf{MaxPoolSample}$, $\mathsf{LinfSample}$ offers a significant advantage in terms of overhead.

\subsection{Byzantine Robustness of FLURP}\label{roubustness_iid}

\begin{figure}[!htbp]
	\centering
        \captionsetup[subfigure]{font=scriptsize} 
        \subfloat[Byzantine resilience under Task 1 (CIFAR-10+ResNet10).]{\includegraphics[width=0.48\textwidth]{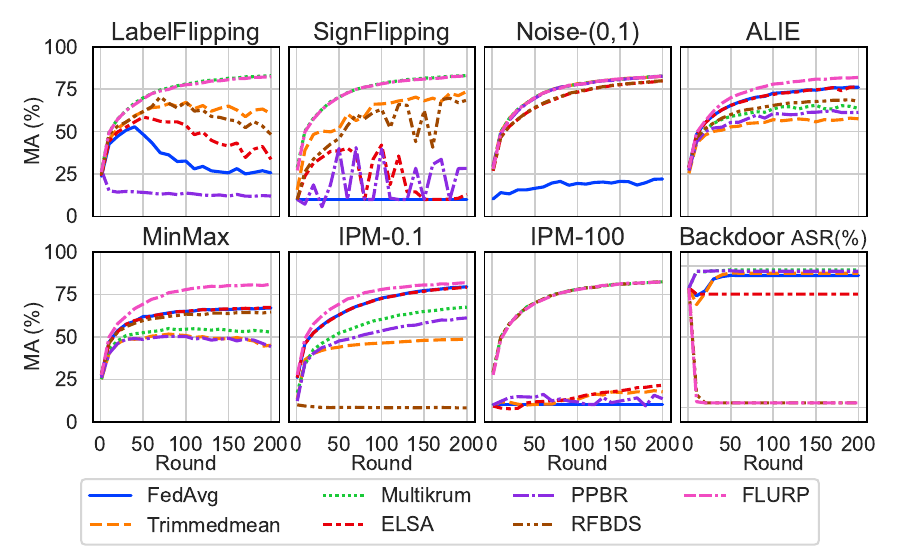}\label{subfig:Cifar10_8_subfigs}}
        \hspace{-0.2em}
        \subfloat[Byzantine resilience under Task 2 (ImageNet-12+MobilNet-V2).]{\includegraphics[width=0.48\textwidth]{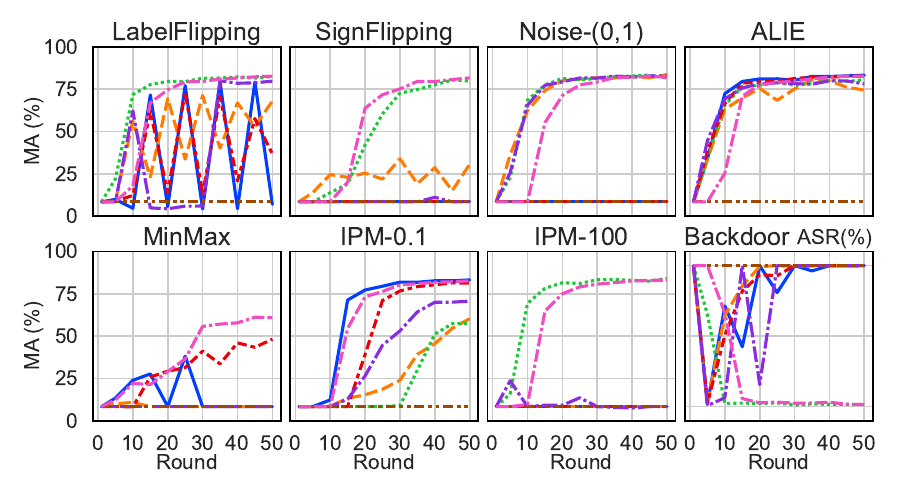}\label{subfig:ImageNet-12_8_subfigs}}
        \hspace{-0.2em}
        \subfloat[Byzantine resilience under Task 3 (AgNews+BiLSTM).]{\includegraphics[width=0.48\textwidth]{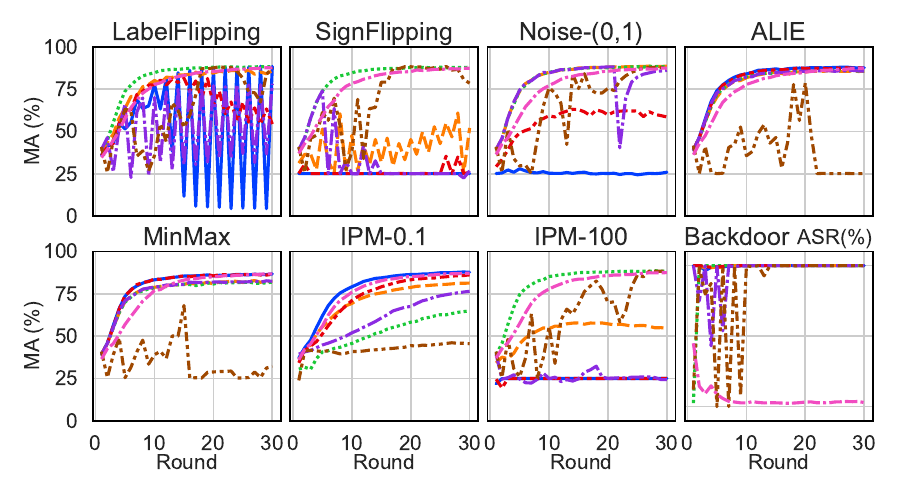}\label{subfig:AgNews_8_subfigs}}
	\caption{Under the scenario with a 40\% malicious client ratio, MA and ASR (\%) of FLURP and the other six comparative schemes against Byzantine attacks across three datasets.}
	\label{fig:EvaluationByzantineattacks}
\end{figure} 

We evaluate the Byzantine resilience of six baselines and FLURP:
\begin{itemize}
    \item \textbf{FedAvg}~\cite{fedavg-pmlr-v54-mcmahan17a}: Aggregates updates by computing their weighted average. 
    \item \textbf{Trimmedmean}~\cite{trimmedmean-pmlr-v80-yin18a}: Removes the highest and lowest 40\% of values in each dimension and computes the average of the remaining values.
    \item \textbf{Multikrum}~\cite{multikrum-NIPS2017_f4b9ec30}: Computes the sum of Euclidean distances to the half nearest neighbors of $\boldsymbol{g}_i$ and uses this sum as a score. A smaller score indicates higher credibility, and the server aggregates updates from the half with the smallest scores. 
    \item \textbf{ELSA}~\cite{ELSA-DBLP:conf/sp/RatheeSWP23}: Clips updates exceeding the ideal norm to the ideal norm, where the ideal norm is the average of the norms of updates from all benign clients.
    \item \textbf{PPBR}~\cite{ppbrf-dongcai2023privacy}: Calculates the sum of cosine similarities to the half nearest neighbors of $\boldsymbol{g}_i$ as a credibility score, then aggregates the half updates with the highest credibility scores.
    \item \textbf{RFBDS}~\cite{RFBDS-chen2023privacy}: Samples updates using $\mathsf{AlignSample}$, clusters them with OPTICS, and clips the updates using the $l_2$ norm.
\end{itemize}

We present the Byzantine resilience of BRAR baselines and FLURP across three tasks in Figure~\ref{fig:EvaluationByzantineattacks}, with all datasets following an IID distribution. 
Only FLURP consistently maintains the highest accuracy, suggesting that the combination of $\mathsf{LinfSample}$ and proximity-based defense effectively distinguishes poisoned updates under these attacks by sampling on $l_{\infty}$.
Furthermore, adversaries utilizing Noise and ALIE attacks are generally mitigated by most BRAR methods except for FedAvg. Notably, MinMax in Task 2 disrupts all BRARs except FLURP and ELSA, while other methods maintain moderate accuracy. We hypothesize that, in Task 2, the model underfits local data due to a low number of local epochs, resulting in higher variance in the updates uploaded each round. This increased variance complicates the detection of poisoned updates. Therefore, we adjust the number of local epochs to 10, leading to an increase in MA for FedAvg and FLURP to 63\% and 78\%, respectively, with the MA of other BRARs ranging between 60\% and 75\%. Notably, FLURP is the only method that effectively resists Backdoor attacks across all tasks. This further illustrates the challenge of detecting backdoor-based poisoned updates in text datasets. Overall, FLURP demonstrates exceptional Byzantine resilience. 

\subsection{Influence of Window Size}\label{sec:windowsize}

\begin{table*}[!htbp]
\caption{The MA and ASR (\%) of FLURP against Byzantine attacks with different sizes of sampling windows.}\label{table:windowsize}
\begin{center} 
\begin{adjustbox}{width=0.99\textwidth}
\begin{tabular}{c|c|S[table-format=2.1]S[table-format=2.1]S[table-format=2.1]S[table-format=2.1]S[table-format=2.1]S[table-format=2.1]S[table-format=2.1]S[table-format=2.1]S[table-format=2.1]S[table-format=2.1]S[table-format=2.1]S[table-format=2.1]S[table-format=2.1]S[table-format=2.1]S[table-format=2.1]S[table-format=2.1]S[table-format=2.1]S[table-format=3.1]cc} 
\toprule
\multirow{2}{*}{Task} & \multirow{2}{*}{Attack} & \multicolumn{20}{c}{The size of the sampling window. ($s = 2^w$)}                                                                                                                                                                                                                                                                                                                                                                                                                                                                                                                                 \\ 
\cline{3-22}
                                                                                           &                         & \multicolumn{1}{c}{$2^4$}    & \multicolumn{1}{c}{$2^5$}    & \multicolumn{1}{c}{$2^6$}   & \multicolumn{1}{c}{$2^7$}   & \multicolumn{1}{c}{$2^8$}   & \multicolumn{1}{c}{$2^9$}   & \multicolumn{1}{c}{$2^{10}$}  & \multicolumn{1}{c}{$2^{11}$}  & \multicolumn{1}{c}{$2^{12}$}  & \multicolumn{1}{c}{$2^{13}$}  & \multicolumn{1}{c}{$2^{14}$}  & \multicolumn{1}{c}{$2^{15}$}  & \multicolumn{1}{c}{$2^{16}$}  & \multicolumn{1}{c}{$2^{17}$}  & \multicolumn{1}{c}{$2^{18}$}  & \multicolumn{1}{c}{$2^{19}$}   & \multicolumn{1}{c}{$2^{20}$}   & \multicolumn{1}{c}{$2^{21}$}   & \multicolumn{1}{c}{$2^{22}$}                      & $2^{23}$                        \\ 
\hline
\multirow{8}{*}{\begin{tabular}[c]{@{}c@{}}Task 1\\(CIFAR-10+ResNet10)\end{tabular}}       & LabelFlipping           & 82.8                     & 82.4                     & 82.2                    & 82.3                    & 82.4                    & 82.5                    & 82.8                    & 82.7                    & 82.3                    & 82.7                    & 82.7                    & 82.7                    & 82.4                    & 82.4                    & 82.8                    & 82.3                     & 82.4                     & 82.7                     & \multicolumn{1}{l}{82.5} & \multicolumn{1}{l}{72.6}  \\
                                                                                           & SignFlipping            & 82.5                     & 82.6                     & 82.6                    & 82.7                    & 82.6                    & 82.6                    & 82.5                    & 82.9                    & 82.9                    & 82.6                    & 82.7                    & 82.5                    & 82.7                    & 82.8                    & 82.4                    & 82.8                     & 82.7                     & 82.6                     & \multicolumn{1}{l}{82.7} & \multicolumn{1}{l}{83.2}  \\
                                                                                           & Noise-(0,1)             & 82.3                     & 82.5                     & 82.3                    & 82.5                    & 82.7                    & 82.8                    & 82.3                    & 82.7                    & 82.8                    & 82.7                    & 82.1                    & 83.0                    & 82.9                    & 82.5                    & 82.6                    & 82.7                     & 82.9                     & 82.6                     & \multicolumn{1}{l}{82.8} & \multicolumn{1}{l}{82.9}  \\
                                                                                           & ALIE                    & 81.5                     & 81.7                     & 81.5                    & 81.8                    & 81.9                    & 82.0                    & 82.0                    & 82.1                    & 81.9                    & 82.0                    & 81.8                    & 82.1                    & 82.2                    & 82.0                    & 82.2                    & 82.1                     & 82.0                     & 81.3                     & \multicolumn{1}{l}{81.2} & \multicolumn{1}{l}{80.9}  \\
                                                                                           & MinMax                  & 61.6                     & 60.7                     & 61.4                    & 57.4                    & 73.4                    & 80.4                    & 80.6                    & 81.0                    & 81.1                    & 81.4                    & 81.7                    & 81.3                    & 81.5                    & 81.0                    & 73.9                    & 72.4                     & 70.9                     & 69.2                     & \multicolumn{1}{l}{71.5} & \multicolumn{1}{l}{68.1}  \\
                                                                                           & IPM-0.1                 & 81.8                     & 82.0                     & 81.5                    & 81.7                    & 81.9                    & 82.0                    & 81.8                    & 82.0                    & 82.0                    & 81.8                    & 81.6                    & 82.0                    & 81.8                    & 82.3                    & 82.4                    & 82.2                     & 82.3                     & 82.4                     & \multicolumn{1}{l}{82.0} & \multicolumn{1}{l}{82.4}  \\
                                                                                           & IPM-100                 & 82.4                     & 82.3                     & 82.3                    & 82.3                    & 82.5                    & 82.4                    & 82.4                    & 82.3                    & 82.4                    & 82.5                    & 82.3                    & 82.4                    & 82.2                    & 82.6                    & 82.7                    & 82.4                     & 82.3                     & 82.5                     & \multicolumn{1}{l}{82.6} & \multicolumn{1}{l}{82.7}  \\
                                                                                           & Backdoor(ASR)           & \multicolumn{1}{c}{25.2} & \multicolumn{1}{c}{14.2} & \multicolumn{1}{c}{3.6} & \multicolumn{1}{c}{8.5} & \multicolumn{1}{c}{3.7} & \multicolumn{1}{c}{3.7} & \multicolumn{1}{c}{3.3} & \multicolumn{1}{c}{3.7} & \multicolumn{1}{c}{4.1} & \multicolumn{1}{c}{3.5} & \multicolumn{1}{c}{3.4} & \multicolumn{1}{c}{3.9} & \multicolumn{1}{c}{3.6} & \multicolumn{1}{c}{4.0} & \multicolumn{1}{c}{4.3} & \multicolumn{1}{c}{3.8}  & \multicolumn{1}{c}{5.1}  & \multicolumn{1}{c}{3.8}  & 7.0                      & 48.0                      \\ 
\hline
\multirow{8}{*}{\begin{tabular}[c]{@{}c@{}}Task 2\\(ImageNet-12+MobilNet-V2)\end{tabular}} & LabelFlipping           & 82.5                     & 82.3                     & 82.7                    & 82.3                    & 82.1                    & 82.5                    & 82.5                    & 82.7                    & 82.3                    & 82.3                    & 82.1                    & 82.3                    & 82.7                    & 82.3                    & 82.6                    & 82.3                     & 80.1                     & 67.9                     & 60.6                     & -                         \\
                                                                                           & SignFlipping            & 81.8                     & 81.4                     & 81.2                    & 81.6                    & 81.5                    & 81.6                    & 81.1                    & 81.6                    & 81.4                    & 81.6                    & 81.2                    & 81.6                    & 81.5                    & 81.1                    & 81.3                    & 81.5                     & 81.6                     & 81.4                     & 81.6                     & -                         \\
                                                                                           & Noise-(0,1)             & 82.6                     & 82.2                     & 82.0                    & 82.4                    & 82.1                    & 82.4                    & 82.3                    & 82.7                    & 82.2                    & 82.0                    & 82.6                    & 82.4                    & 82.7                    & 82.5                    & 82.2                    & 82.2                     & 82.4                     & 82.6                     & 82.4                     & -                         \\
                                                                                           & ALIE                    & 82.2                     & 82.3                     & 82.4                    & 82.2                    & 82.7                    & 82.2                    & 82.7                    & 82.5                    & 82.9                    & 82.7                    & 82.5                    & 82.5                    & 82.2                    & 82.7                    & 82.2                    & 82.5                     & 81.3                     & 80.7                     & 76.1                     & -                         \\
                                                                                           & MinMax                  & 36.9                     & 47.5                     & 44.9                    & 50.9                    & 55.4                    & 61.4                    & 61.7                    & 60.6                    & 60.9                    & 61.9                    & 58.7                    & 60.6                    & 58.4                    & 55.9                    & 58.1                    & 49.9                     & 42.0                     & 38.2                     & 36.9                     & -                         \\
                                                                                           & IPM-0.1                 & 82.7                     & 82.4                     & 82.4                    & 82.1                    & 82.7                    & 82.3                    & 82.3                    & 82.4                    & 82.7                    & 82.4                    & 82.1                    & 82.2                    & 82.3                    & 82.4                    & 82.4                    & 82.3                     & 82.3                     & 81.3                     & 82.3                     & -                         \\
                                                                                           & IPM-100                 & 83.1                     & 83.2                     & 83.0                    & 83.1                    & 83.2                    & 83.1                    & 83.0                    & 83.2                    & 83.0                    & 83.2                    & 83.0                    & 83.2                    & 83.0                    & 83.3                    & 83.0                    & 83.1                     & 83.2                     & 83.3                     & 83.4                     & -                         \\
                                                                                           & Backdoor(ASR)           & 31.5                     & 14.1                     & 8.5                     & 10.2                    & 2.4                     & 1.4                     & 1.4                     & 1.7                     & 1.5                     & 2.8                     & 2.3                     & 2.7                     & 2.9                     & 2.8                     & 5.5                     & 7.6                      & 10.5                     & 97.0                     & 99.0                     & -                         \\ 
\hline
\multirow{8}{*}{\begin{tabular}[c]{@{}c@{}}Task 3\\(AgNews+BiLSTM)\end{tabular}}           & LabelFlipping           & 87.3                     & 87.1                     & 86.9                    & 87.0                    & 87.6                    & 86.8                    & 87.1                    & 87.3                    & 87.9                    & 87.3                    & 87.6                    & 87.3                    & 87.6                    & 87.3                    & 87.1                    & \multicolumn{1}{c}{84.8} & \multicolumn{1}{c}{82.6} & \multicolumn{1}{c}{65.6} & 62.1                     & 58.0                      \\
                                                                                           & SignFlipping            & 87.3                     & 87.4                     & 87.1                    & 86.6                    & 87.3                    & 87.4                    & 87.3                    & 87.1                    & 87.4                    & 87.1                    & 87.6                    & 87.3                    & 87.4                    & 87.1                    & 87.3                    & \multicolumn{1}{c}{87.1} & \multicolumn{1}{c}{87.4} & \multicolumn{1}{c}{87.6} & 86.4                     & 85.1                      \\
                                                                                           & Noise-(0,1)             & 87.4                     & 87.1                     & 87.8                    & 87.3                    & 87.8                    & 87.2                    & 87.4                    & 87.5                    & 87.3                    & 87.2                    & 87.5                    & 87.1                    & 87.8                    & 87.2                    & 87.5                    & \multicolumn{1}{c}{87.2} & \multicolumn{1}{c}{87.3} & \multicolumn{1}{c}{87.1} & 86.7                     & 86.1                      \\
                                                                                           & ALIE                    & 86.8                     & 86.5                     & 86.6                    & 87.0                    & 87.2                    & 86.8                    & 87.3                    & 87.0                    & 87.2                    & 87.3                    & 87.1                    & 87.1                    & 87.3                    & 87.4                    & 87.1                    & \multicolumn{1}{c}{87.4} & \multicolumn{1}{c}{82.4} & \multicolumn{1}{c}{80.0} & 78.8                     & 76.5                      \\
                                                                                           & MinMax                  & 60.8                     & 71.3                     & 76.9                    & 79.6                    & 81.2                    & 85.8                    & 86.5                    & 86.6                    & 86.4                    & 86.4                    & 86.8                    & 86.6                    & 81.2                    & 78.8                    & 74.0                    & \multicolumn{1}{c}{60.2} & \multicolumn{1}{c}{57.2} & \multicolumn{1}{c}{55.4} & 50.0                     & 48.8                      \\
                                                                                           & IPM-0.1                 & 85.5                     & 86.1                     & 87.4                    & 87.2                    & 87.7                    & 87.4                    & 87.1                    & 87.6                    & 87.1                    & 87.2                    & 87.5                    & 87.2                    & 87.4                    & 87.6                    & 87.1                    & \multicolumn{1}{c}{87.4} & \multicolumn{1}{c}{87.7} & \multicolumn{1}{c}{87.4} & 87.7                     & 87.1                      \\
                                                                                           & IPM-100                 & 87.6                     & 87.9                     & 87.6                    & 87.4                    & 87.8                    & 87.8                    & 87.5                    & 87.4                    & 87.6                    & 87.7                    & 87.9                    & 87.7                    & 87.5                    & 87.3                    & 87.9                    & \multicolumn{1}{c}{87.8} & \multicolumn{1}{c}{87.5} & \multicolumn{1}{c}{87.8} & 87.6                     & 87.4                      \\
                                                                                           & Backdoor(ASR)           & 32.7                     & 24.2                     & 19.6                    & 4.7                     & 3.2                     & 2.9                     & 3.4                     & 3.9                     & 3.7                     & 4.2                     & 3.8                     & 4.3                     & 3.5                     & 3.4                     & 6.5                     & \multicolumn{1}{c}{8.9}  & \multicolumn{1}{c}{15.5} & \multicolumn{1}{c}{22.9} & 42.8                     & 59.1                      \\
\bottomrule
\end{tabular}
\end{adjustbox}
\end{center}
\end{table*}

To ensure FLURP consistently filters the aforementioned attacks while training global models across diverse tasks, selecting the optimal sampling window size is crucial. Intuitively, when the sampling window size, $s = 2^w$, in $\mathsf{LinfSample}$ is too small, detecting subtle changes in poisoned model updates can be challenging. Conversely, an excessively large window size may result in significant information loss. We evaluate MA of FLURP across different window sizes under the eight aforementioned attacks, with the results presented in Table~\ref{table:windowsize}. 

Our analysis reveals that the attacks most sensitive to the window size $s$ are LabelFlipping, ALIE, MinMax, and Backdoor. Among these, FLURP defends against LabelFlipping and ALIE with smaller window sizes, while it requires moderate window sizes to defend against MinMax and Backdoor. Remarkably, FLURP continues to offer superior protection, even against the more complex attack types, while ensuring stability across varying settings. The remaining four simpler attacks show little sensitivity to the window size. Based on empirical analysis of the experimental results, a window size of $2^{11}$, $2^{12}$, $2^{13}$, or $2^{14}$ is deemed optimal for these million-sized models, striking the best balance between overhead and defense.

\subsection{Impact of Client Data Distribution}\label{sec:roubustness_dirichlet}

\begin{figure}[!htbp]
\centering
\includegraphics[width=0.48\textwidth]{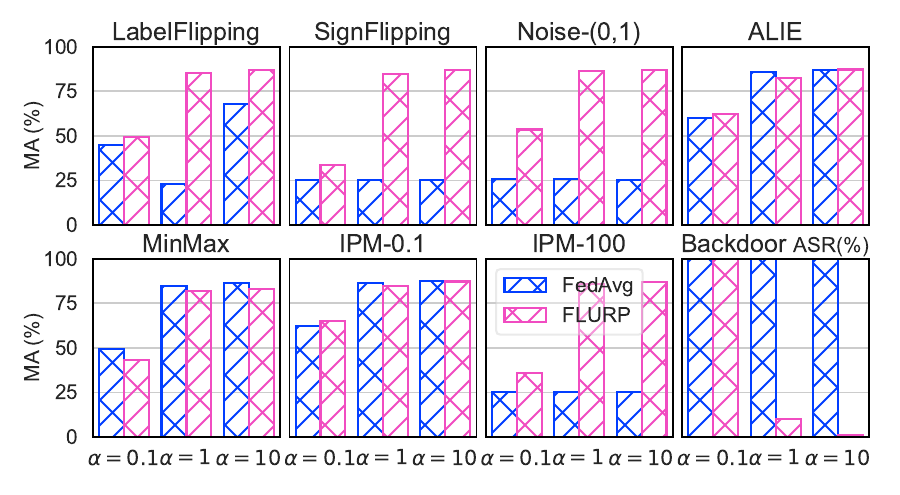}
\caption{The MA and ASR (\%) of FedAvg and FLURP under Task 3 (AgNews + Bi-LSTM) across three different non-IID distributions.}
\label{fig:dirichlet}
\end{figure}

We use the Dirichlet distribution, as described in \cite{lin2020ensemble_dirichlet}, to partition the training set. A smaller $\alpha$ indicates a greater degree of non-IIDness, with clients likely possessing samples predominantly from a single category. Conversely, a larger $\alpha$ results in more uniform data distributions across clients. We test three values of $\alpha$: 10, 1, and 0.1, with experimental results shown in Figure~\ref{fig:dirichlet}. When $\alpha = 10$, FLURP performs exceptionally well, demonstrating the highest robustness across all attack scenarios. At $\alpha = 1$, FLURP continues to exhibit superior performance compared to FedAvg in most scenarios, effectively defending against a variety of attacks. However, when $\alpha$ is reduced to 0.1, the increased divergence between model updates from benign clients introduces some instability. FedAvg also experiences significant degradation, and FLURP struggles to maintain optimal performance. In particular, FLURP faces challenges with MinMax and Backdoor attacks. Despite these challenges under $\alpha = 0.1$, FLURP remains a strong contender, especially when $\alpha$ is high. As $\alpha$ decreases, filtering out malicious clients becomes harder, but FLURP still provides robust defense in most situations.

\subsection{Impact of the Proportion of Malicious Clients.}\label{sec:proportion}

\begin{table}
\centering
\caption{MA and ASR (\%) of FLURP under varying proportions of malicious clients.}
\label{tab:varyingmalicious}
\begin{tabular}{c|ccccc} 
\toprule
\multirow{2}{*}{Attack} & \multicolumn{5}{c}{Malicious Clients} \\ 
\cline{2-6}
                                 & 0\%   & 10\%  & 20\%  & 30\%  & 40\%  \\ 
\hline
LabelFlipping (MA (\%))           & 89.1  & 89.2  & 88.8  & 88.9  & 87.9  \\
ALIE (MA (\%))                   & 88.6  & 88.5  & 88.0  & 88.2  & 87.2  \\
MinMax (MA (\%))                 & 88.9  & 88.3  & 88.6  & 87.3  & 86.4  \\
Backdoor (ASR (\%))               &  0.7  &  1.9  &  2.5  &  3.8  &  3.7  \\
\bottomrule
\end{tabular}
\end{table}

The experimental results in Table \ref{tab:varyingmalicious} demonstrate FLURP’s stable resilience against varying proportions of malicious clients. The MA of FLURP remains largely unaffected across attack types, with only a slight decline as the proportion of malicious clients increases. For example, the MA for LabelFlipping decreases from 89.1\% to 87.9\% as the proportion of malicious clients rises from 0\% to 40\%, indicating FLURP’s ability to maintain high accuracy despite attacks. Similarly, MA for ALIE and MinMax shows minimal degradation. Notably, the ASR for the Backdoor attack remain very low, starting at 0.7\% and only rising slightly to 3.7\% at 40\% malicious clients. This confirms that FLURP effectively mitigates the impact of backdoor attacks, even with a higher proportion of malicious clients, making it a highly effective defense for FL systems.

\subsection{Adaptive Attack against FLURP}\label{sec:Adaptive}

Considering the potential for collusion between the server and certain clients, we introduce a stronger adversary, Adaptive-FLURP, which has knowledge of both the FLURP aggregation rule and the updates from all benign clients. Inspired by \cite{minmax-shejwalkar2021manipulating}, this attack aims to maximize the number of poisoned updates selected by the aggregation rule, thereby increasing the deviation in the global model. The optimization problem for Adaptive-FLURP is formulated as follows:

\newcommand{\argmax}{\mathop{\text{argmax}}}

\begin{align}\label{eq:adaptive} 
    \argmax_{\gamma} m &= \left| \left\{ \boldsymbol{g}_i \mid \boldsymbol{g}_i \in \left[ \boldsymbol{g}_i \right]_{i \in \mathcal{C}'} \right\} \cap \left\{ \boldsymbol{g}_i \mid \boldsymbol{g}_i \in \left[ \boldsymbol{g}_i \right]_{i \in \mathcal{Y}_{\text{FLURP}}} \right\} \right| \notag \\
    \boldsymbol{g}_i &= \boldsymbol{\mu} + \gamma \cdot \boldsymbol{\sigma}, \quad \text{for} \  i \in \mathcal{C}'
\end{align}

where $\mathcal{C}'$ denotes malicious clients, $\mathcal{Y}_{\text{FLURP}}$ is the set of FLURP-selected clients, $\gamma$ is the scaling factor. Adaptive-FLURP solves \ref{eq:adaptive} using the method from \cite{minmax-shejwalkar2021manipulating} to find the optimal $\gamma$. We evaluate FLURP's defense against Adaptive-FLURP on ImageNet-12 and AgNews across IID and three non-IID distributions as described in Section~\ref{sec:modelsetting}. Metrics include MA, the proportion of malicious clients in $\mathcal{Y}_{\text{FLURP}}$ (denoted as $p$), and the poisoned update shift $\varepsilon$ (defined as $\varepsilon = \frac{|\gamma \cdot \boldsymbol{\sigma}|_2}{|\boldsymbol{\mu}|_2}$). The results in Figure~\ref{fig:adaptive} show that while Adaptive-FLURP is effective across various data distributions and allows some poisoned updates to bypass FLURP's detection, it must choose a smaller $\gamma$ to avoid detection by $\mathsf{LinfSample}$, resulting in minimal deviation in the global model. In these cases, the impact on the model's MA is negligible. When $\alpha = 0.1$, the distance between benign updates increases, allowing Adaptive-FLURP to inject more deviation through poisoned updates. Nevertheless, FLURP maintains model robustness. These experimental results demonstrate FLURP's reliability as a defense mechanism even against adaptive adversaries.

\begin{figure}[!htbp]
\centering
\includegraphics[width=0.48\textwidth]{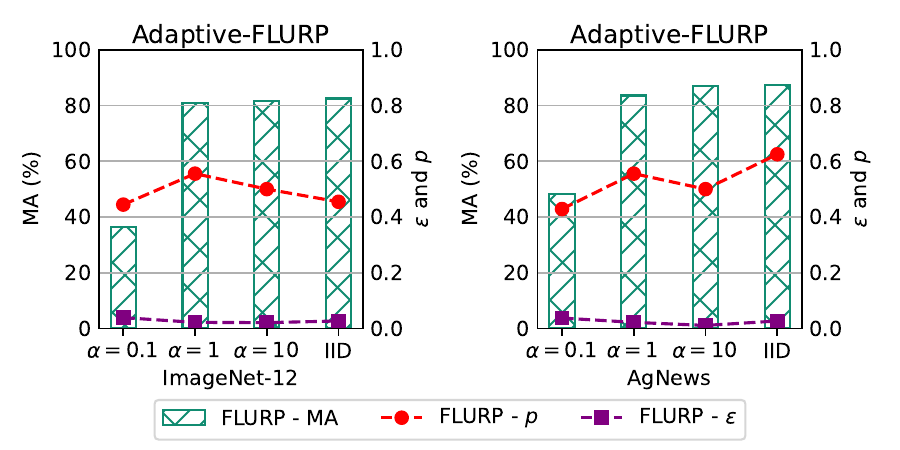}
\caption{Robustness of FLURP under Adaptive-FLURP Attack.}
\label{fig:adaptive}
\end{figure}

\subsection{Efficiency of Computing Medians}

\begin{figure}[!htbp]
\centering
\includegraphics[width=0.48\textwidth]{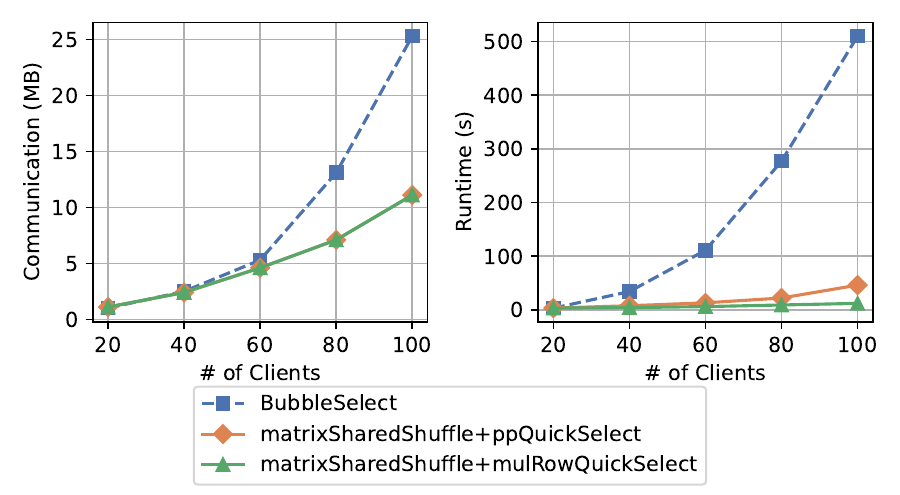}
\caption{Overheads of computing medians for servers.}
\label{fig:comm}
\end{figure}

 To evaluate FLURP's SMPC overhead, we compare the overhead of computing medians using the $\mathsf{matrixSharedShuffle+mulRowQuickSelect}$ protocol of FLURP with two baselines. The $\mathsf{BubbleSelect}$ protocol \cite{asscryptflow2-rathee2020cryptflow2} identifies the $m/2$ largest values in each row using privacy-preserving max-pooling and bubble sort. The $\mathsf{matrixSharedShuffle+ppQuickSelect}$ protocol uses the same shuffling algorithm as FLURP but applies the selection method from \cite{pbfl-li2023pbfl} to compute the shared median row by row. The overhead of these protocols is primarily determined by the number of clients, i.e., the size of the distance matrix. Thus, we analyze the communication and runtime overhead for client counts of 20, 40, 60, 80, and 100.

The comparison complexity of $\mathsf{BubbleSelect}$ is $O(m^3)$, while both $\mathsf{matrixSharedShuffle+mulRowQuickSelect}$ and $\mathsf{matrixSharedShuffle+ppQuickSelect}$ have a comparison complexity of $O(m)$. The $\mathsf{mulRowQuickSelect}$ protocol performs row-wise packing and comparisons, resulting in a worst-case complexity of $O(m^2)$ and an average complexity of $O(m)$. In contrast, $\mathsf{ppQuickSelect}$ has an average complexity of $O(m^2)$.

As shown in left side of Figure~\ref{fig:comm}, the communication overhead of all three protocols is similar with 20 and 40 clients. However, as the number of clients increases, the overhead of the $\mathsf{BubbleSelect}$ protocol, despite not requiring row shuffling, grows steeply, while the other two protocols exhibit a more gradual increase. At 100 clients, the communication overhead of $\mathsf{BubbleSelect}$ is approximately twice that of the other two protocols. 

Regarding runtime (right side of Figure~\ref{fig:comm}), all protocols perform similarly with 20 clients. As client count increases, $\mathsf{mulRowQuickSelect}$ shows the lowest runtime overhead and the slowest increase compared to the others. Overall, FLURP shows linear communication and runtime SMPC overhead, demonstrating its good scalability.

\section{Conclusion}

This study introduces the FLURP framework, which significantly advances data security, privacy, and access control within distributed environments. By integrating LURs and a novel, privacy-preserving proximity-based defense mechanism, FLURP effectively mitigates the risks posed by Byzantine adversaries. Furthermore, FLURP reduces both the computational and communication overhead in SMPC, making it a more efficient solution in FL systems. As a result, FLURP improves the overall robustness of FL against security threats, ensuring safer data management practices for AI applications.

In future work, we plan to enhance the performance of FLURP under more complex non-IID data distributions by exploring the integration of other techniques to strengthen its defense, such as data augmentation methods to reduce non-IIDness. We also plan to defend against adversaries capable of forging LURs by transforming $\mathsf{LinfSample}$ into an SMPC protocol that operates on two servers.

\bibliographystyle{unsrt}
\bibliography{sample-base}



\begin{IEEEbiography}[{\includegraphics[width=1in,height=1.25in,clip,keepaspectratio]{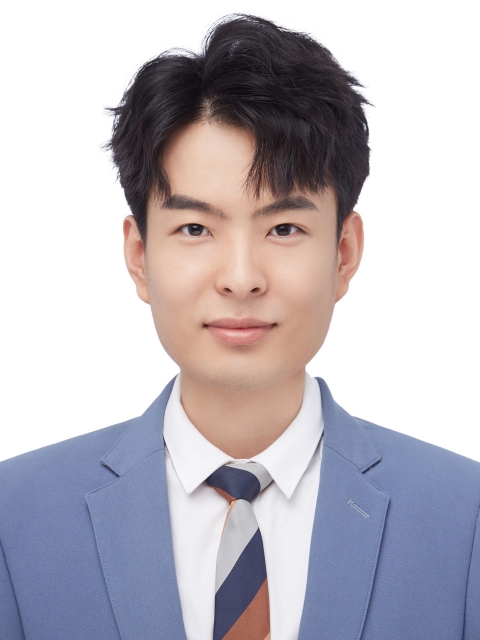}}]{Wenjie Li} received his M.S. degree from the School of Cyber Security and Computer, Hebei University of China in 2021. Currently, he is a Ph.D. student at the School of Cyber Engineering, Xidian University, China, and also a visiting student at the College of Computing and Data Science (CCDS) at Nanyang Technological University (NTU). He is working on cryptography, secure aggregation in federated learning and privacy-preserving machine learning.

\end{IEEEbiography}

\begin{IEEEbiography}[{\includegraphics[width=1in,height=1.25in,clip,keepaspectratio]{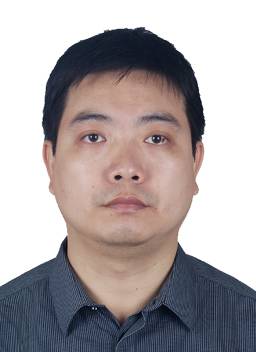}}]{Kai Fan} received his B.S., M.S. and Ph.D. degrees from Xidian University, P.R.China, in 2002, 2005 and 2007, respectively, in Telecommunication Engineering, Cryptography and Telecommunication and Information System. He is working as a professor in State Key Laboratory of Integrated Service Networks at Xidian University. He published over 70 papers in journals and conferences. He has managed 5 national research projects. His research interests include IoT security and information security. 
\end{IEEEbiography}

\begin{IEEEbiography}[{\includegraphics[width=1in,height=1.25in,clip,keepaspectratio]{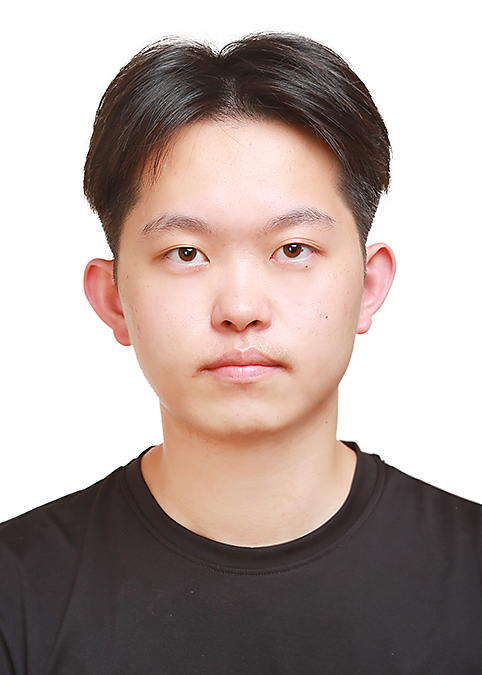}}]{Jingyuan Zhang} received the B.Eng. degree in cyber space security from the School of Cyber Science and Engineering, Wuhan University. He is currently pursuing a Master of Science in Artificial Intelligence at the College of Computing and Data Science (CCDS) at Nanyang Technological University (NTU). His research interests include federated learning and federated transfer learning.
\end{IEEEbiography}

\begin{IEEEbiography}[{\includegraphics[width=1in,height=1.25in,clip,keepaspectratio]{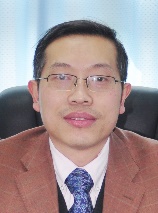}}]{Hui Li} was born in 1968 in Shaanxi Province of China. In 1990, he received his B. S. degree in radio electronics from Fudan University. In 1993, and 1998, he received his M.S. degree and Ph.D. degree in telecommunications and information system from Xidian University respectively. He is now a professor of Xidian University. His research interests include network and information security.
\end{IEEEbiography}

\begin{IEEEbiography}[{\includegraphics[width=1in,height=1.25in,clip,keepaspectratio]{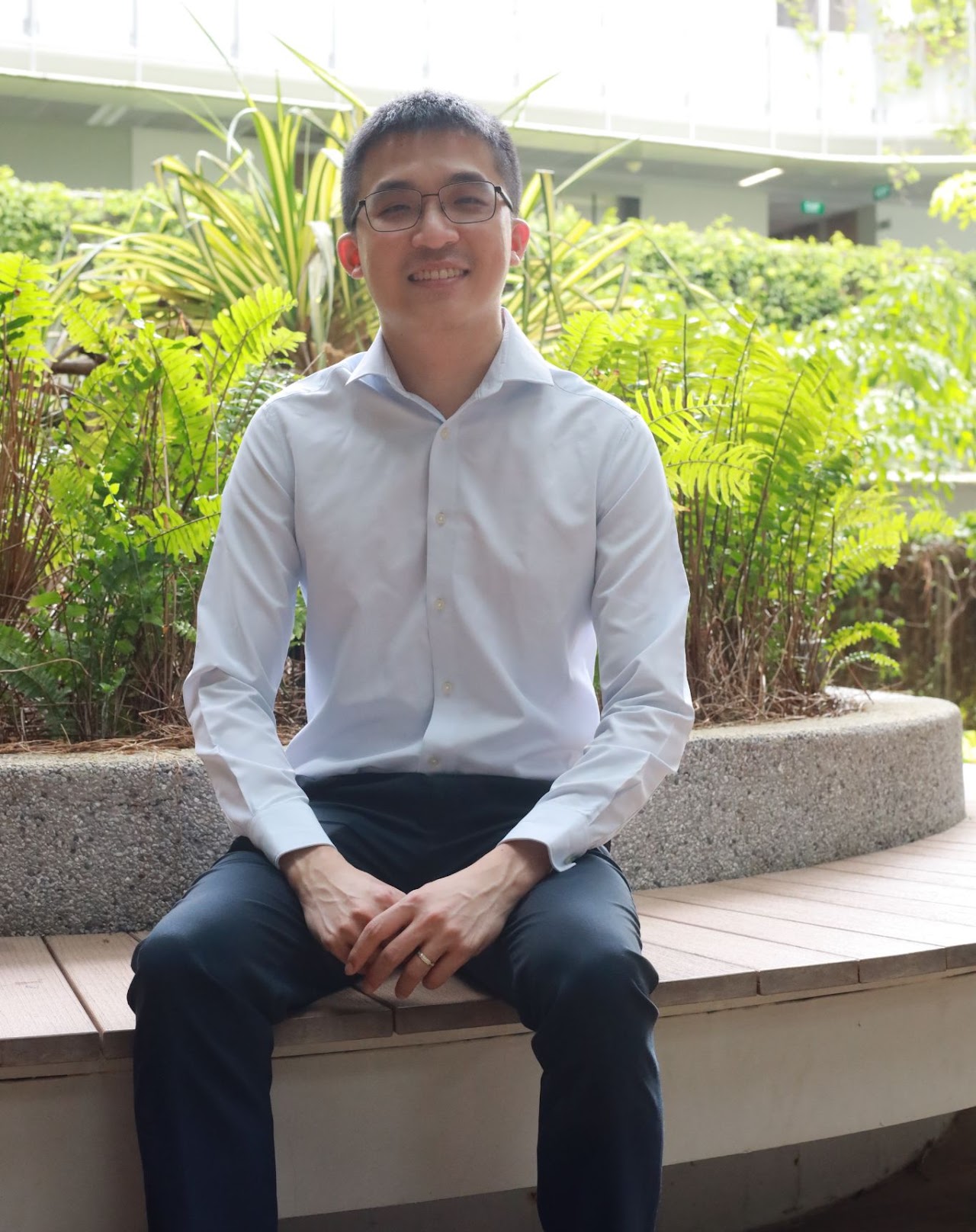}}]{Wei Yang Bryan Lim} is currently an Assistant Professor at the College of Computing and Data Science (CCDS), Nanyang Technological University (NTU), Singapore. Previously, he was Wallenberg-NTU Presidential Postdoctoral Fellow. In 2022, he earned his PhD from NTU under the Alibaba PhD Talent Programme and was affiliated with the CityBrain team of DAMO academy. His doctoral efforts earned him accolades such as the ``Most Promising Industrial Postgraduate Programme Student" award and the IEEE Technical Community on Scalable Computing (TCSC) Outstanding PhD Dissertation Award. He has also won the best paper awards, notably from the IEEE Wireless Communications and Networking Conference (WCNC) and the IEEE Asia Pacific Board. He serves on the Technical Programme Committee for FL workshops at flagship conferences (AAAI-FL, IJCAI-FL) and is a review board member for reputable journals like the IEEE Transactions on Parallel and Distributed Systems. In 2023, he co-edited a submission on ``Requirements and Design Criteria for Sustainable Metaverse Systems" for the International Telecommunication Union (ITU). He has also been a visiting scholar at various institutions such as the University of Tokyo, KTH Royal Institute of Technology, and the University of Sydney.
\end{IEEEbiography}

\begin{IEEEbiography}[{\includegraphics[width=1in,height=1.25in,clip,keepaspectratio]{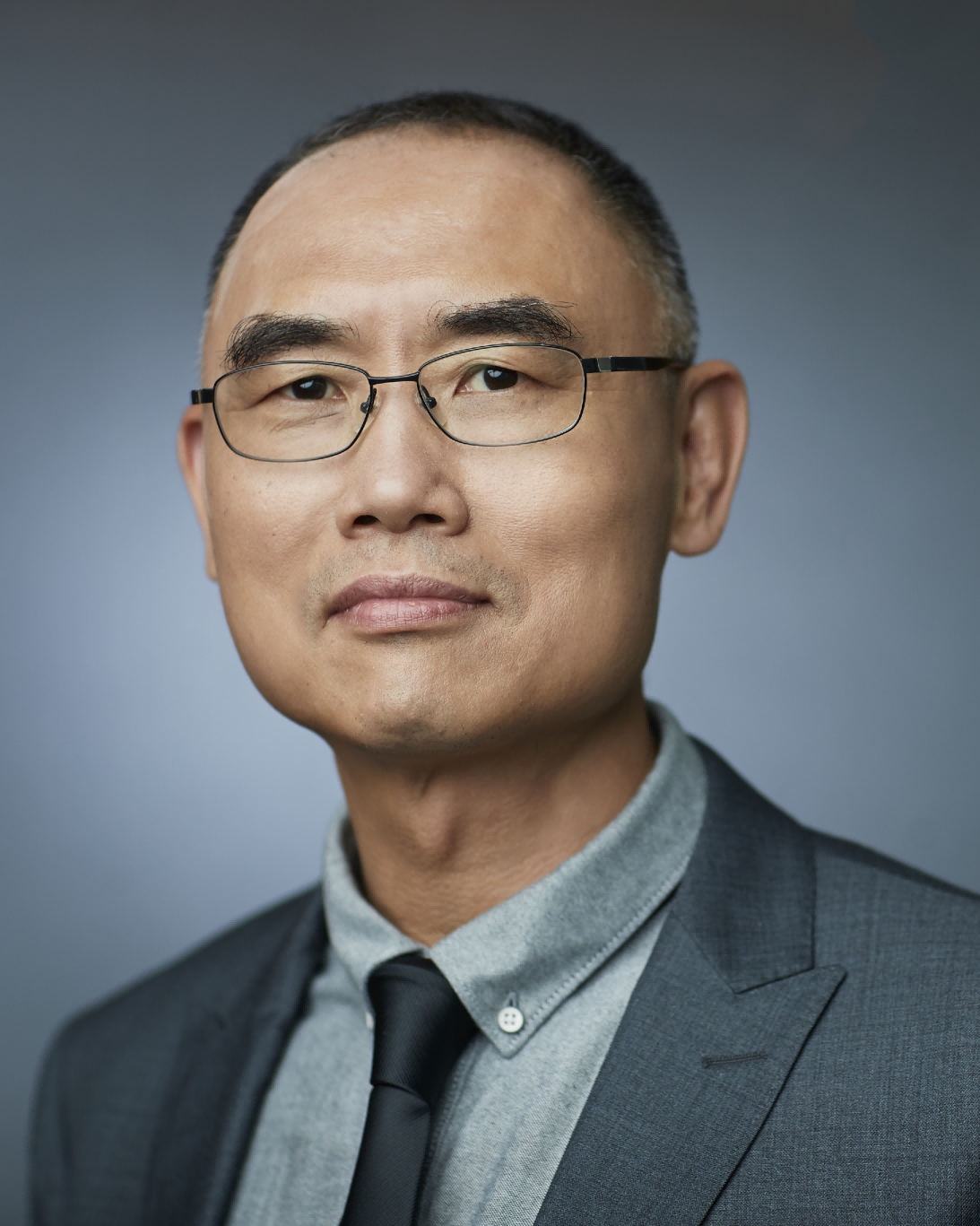}}]{Qiang Yang} is a Fellow of Canadian Academy of Engineering (CAE) and Royal Society of Canada (RSC), Chief Artificial Intelligence Officer of WeBank, a Chair Professor of Computer Science and Engineering Department at Hong Kong University of Science and Technology (HKUST). He is theConference Chair of AAAI-21, the Honorary Vice President of Chinese Association for Artificial Intelligence(CAAI) , the President of Hong Kong Society of Artificial Intelligence and Robotics (HKSAIR)and the President of Investment Technology League (ITL). He is a fellow of AAAI, ACM, CAAI, IEEE, IAPR, AAAS. He was the Founding Editor in Chief of the ACM Transactions on Intelligent Systems and Technology (ACM TIST) and the Founding Editor in Chief of IEEE Transactions on Big Data (IEEE TBD). He received the ACM SIGKDD Distinguished Service Award in 2017. He had been the Founding Director of the Huawei’s Noah’s Ark Research Lab between 2012 and 2015, the Founding Director of HKUST’s Big Data Institute, the Founder of 4Paradigm and the President of IJCAI (2017-2019). His research interests are artificial intelligence, machine learning, data mining and planning.
\end{IEEEbiography}

\vfill

\end{document}